\documentclass[submission,copyright,creativecommons]{eptcs}
\providecommand{\event}{LSFA 2024} 

\usepackage{iftex}
\usepackage{amsthm}
\newtheorem{theorem}{Theorem}
\newtheorem{definition}{Definition}

\newtheorem{lemma}{Lemma}
\newtheorem{remark}{Remark}
\newtheorem{corollary}{Corollary}

\usepackage{latexsym,relsize}
\usepackage{amsmath,amsfonts,amssymb}
\usepackage{mathtools}
\usepackage{algorithm}
\usepackage{algpseudocode}
\usepackage{setspace}
\usepackage{multicol}
\usepackage{longtable}
\usepackage{comment}
\usepackage{color}
\usepackage{bussproofs}
\EnableBpAbbreviations

\usepackage{tikz}
\usepackage{colonequals}

\usepackage{multirow}

\usepackage{multicol}

\usepackage{array}
\newcolumntype{C}[1]{>{\centering\arraybackslash}p{#1}}
\newcolumntype{L}[1]{>{\arraybackslash}p{#1}}

\def\mc{\multicolumn}

\newcommand{\fns}{\footnotesize}

\renewcommand{\emph}{\textbf}

\newcommand{\Prop}{\mathsf{Prop}}

\newcommand{\val}[1]{[\![{#1}]\!]}
\newcommand{\descr}[1]{(\![{#1}]\!)}
\renewcommand{\phi}{\varphi}

\newcommand{\Diamondblack}{\blacklozenge}

\newcommand{\fakeparagraph}[1]{

\textit{#1} \ \ }

\newcommand{\ceq}{\colonequals}

\ifpdf
  \usepackage{underscore}         
  \usepackage[T1]{fontenc}        
\else
  \usepackage{breakurl}           
\fi

\title{Query Answering in Lattice-based Description Logic}
\author{Krishna Manoorkar
\institute{School of Business and Economics\\
Vrije Universiteit Amsterdam\thanks{Krishna Manoorkar is supported by the NWO grant KIVI.2019.001 awarded to Alessandra Palmigiano.}\\
Amsterdam, the Netherlands}
\email{k.b.manoorkar@vu.nl}
\and
 Ruoding Wang
\institute{School of Business and Economics\\Vrije Universiteit Amsterdam\thanks{Ruoding Wang is supported by the China Scholarship Council No.202206310072.}\\
Amsterdam, the Netherlands}
\institute{Department of Philosophy\\Xiamen University\\ Xiamen, China}
\email{r.wang2@vu.nl}
}
\def\titlerunning{Query Answering in Lattice-based Description Logic}
\def\authorrunning{K. B. Manoorkar \& R. Wang}
\begin{document}
\maketitle

\begin{abstract}
Recently, the description logic LE-$\mathcal{ALC}$ was introduced for reasoning in the semantic environment of the enriched formal contexts, and a tableaux algorithm was developed for checking the consistency of ABoxes in this logic \cite{van2023old,van2023non}. In this paper, we study the ontology-mediated query answering in LE-$\mathcal{ALC}$. In particular, we show that several different types of queries can be answered efficiently for LE-$\mathcal{ALC}$  knowledge bases with acyclic TBoxes using our tableaux algorithm directly or by extending it with some additional rules. 
\end{abstract}

\section{Introduction}\label{Sec:Introduction}
Description logic (DL) \cite{DLhandbook} is a class of logical formalisms, rooted in classical first-order logic, widely used in Knowledge Representation and Reasoning to articulate and infer relationships among pertinent concepts within a specified application domain. It is widely utilized across various fields such as the semantic web \cite{horrocks2002daml+oil,baader2005description}, ontologies \cite{staab2010handbook}, and software engineering \cite{berardi2005reasoning}. Description logic offers solutions to diverse reasoning tasks arising from a knowledge base. Among the notable reasoning services offered by description logic is ontology-mediated query answering, which involves answering queries based on a given knowledge base \cite[chpater 7]{DLhandbook}. 

In \cite{van2023non}\footnote{We noticed a  mistake in the proof of termination and $I$-compatibility in an earlier version of this paper \cite{van2023old} in which concepts $\top$ and $\bot$ were included as concept names. In the updated version \cite{van2023non} we prove that the result holds in the restriction which does not contain $\top$ and $\bot$ in the language of concept names. In this paper, we work with the restricted language as in \cite{van2023non}.}, a two-sorted {\em lattice-based description logic} LE-$\mathcal{ALC}$\footnote{Even though concept names in  LE-$\mathcal{ALC}$ do not contain negation, we still refer to this description logic as LE-$\mathcal{ALC}$ rather than LE-$\mathcal{ALE}$, as negation on ABox terms is included in the description logic language.} was introduced based on non-distributive modal logic, with semantics grounded in an enriched formal context~\cite{conradie2016categories,conradie2017toward}. LE-$\mathcal{ALC}$ provides a natural description logic to reason about formal concepts (or categories) arising from formal contexts in Formal Concept Analysis (FCA) \cite{ganter1997applied,ganter2012formal}. The logic LE-$\mathcal{ALC}$ has the same relationship with non-distributive modal logic and its semantics based on formal contexts as the relationship between $\mathcal{ALC}$ and the classical normal modal logic with its Kripke frame semantics. Namely, LE-$\mathcal{ALC}$  facilitates the description of {\em enriched formal contexts}, i.e., formal contexts endowed with additional relations, which give rise to concept lattices extended with normal modal operators. Similarly to the classical modal operators, the `non-distributive' modal operators can be given different interpretations, such as the epistemic operator \cite{conradie2017toward} and the approximation operator \cite{conradie2021rough}. 

In this paper, we adapt and modify the LE-$\mathcal{ALC}$ tableaux algorithm provided in \cite{van2023non} to answer several different types of queries based on LE-$\mathcal{ALC}$  knowledge bases with acyclic TBoxes. We show that for any consistent LE-$\mathcal{ALC}$ ABox $\mathcal{A}$, the model constructed from the tableaux completion of $\mathcal{A}$ is a universal or canonical model for answering different queries like {\em relationship queries} asking if an object and a feature are related, {\em membership queries} asking if an object or a feature belongs to a concept, and {\em subsumption queries} asking if a concept is included in some other concept. This allows us to answer multiple such queries in polynomial time in $|\mathcal{A}|$. We show that it also acts as a universal model w.r.t.~negative relational queries, however this is not true for negative membership or subsumption queries.

Finally, we consider separation queries which ask if two objects or features can be distinguished from each other by means of some role (relation). We convert these queries into an equivalent problem of checking the consistency of the given ABox w.r.t.~some extension of LE-$\mathcal{ALC}$ and providing a tableaux algorithm for such extension. This method allows to answer separation queries of different types in polynomial time in $|\mathcal{A}|$. 
 
{\em Structure of the paper.} In Section~\ref{Sec:Preliminaries}, we briefly review non-distributive modal logic and polarity-based semantics, lattice-based description logic LE-$\mathcal{ALC}$, and the tableaux algorithm for checking its ABox consistency. In Section~\ref{sec:query answering}, we demonstrate that the model obtained from the Tableaux Algorithm (Section~\ref{Sec:Preliminaries}) is a universal model for various queries, and we define different types of queries and corresponding algorithms. Section~\ref{sec:examples} provides a specific LE-$\mathcal{ALC}$ knowledge base and illustrates how the algorithms answer the queries discussed earlier. Finally, Section~\ref{Sec: conclusions} summarizes the paper and outlines future directions.

\section{Preliminaries}\label{Sec:Preliminaries}
In this section, we collect preliminaries on non-distributive modal logic and its polarity-based semantics, i.e. semantics based on formal contexts, and the lattice-based description logic LE-$\mathcal{ALC}$ with the tableaux algorithm developed for it in \cite{van2023non}. 

\subsection{Basic non-distributive modal logic and its polarity-based semantics}\label{ssec:LE-logic}

In this section, we briefly introduce the basic non-distributive modal logic and polarity-based semantics for it. It is a member of a family of lattice-based logics, sometimes referred to as {\em LE-logics} (cf.~\cite{conradie2019algorithmic}), which have been studied in the context of a research program on the logical foundations of categorization theory \cite{conradie2016categories,conradie2017toward,conradie2021rough,conradie2020non}. 
Let $\Prop$ be a (countable) set of atomic propositions. The language $\mathcal{L}$ is defined as follows:

{{\centering
  $\varphi \ceq \bot \mid \top \mid p \mid  \varphi \wedge \varphi \mid \varphi \vee \varphi \mid \Box \varphi \mid \Diamond \varphi$,  
\par}}

\noindent where $p\in \Prop$. 
The {\em basic}, or {\em minimal normal} $\mathcal{L}$-{\em logic} is a set $\mathbf{L}$ of sequents $\phi\vdash\psi$, with $\phi,\psi\in\mathcal{L}$, containing the following axioms:
\smallskip

{{\centering
\begin{tabular}{ccccccccccccc}
     $p \vdash p$ & \qquad & $\bot \vdash p$ & \qquad & $p \vdash p \vee q$ & \qquad & $p \wedge q \vdash p$ & \qquad & $\top \vdash \Box\top$ & \qquad & $\Box p \wedge \Box q \vdash \Box(p \wedge q)$
     \\
     & \qquad & $p \vdash \top$ & \qquad & $q \vdash p \vee q$ & \qquad & $p \wedge q \vdash q$ &\qquad &  $\Diamond\bot \vdash \bot$ & \qquad & $\Diamond(p \vee q) \vdash \Diamond p \vee \Diamond q$\\
\end{tabular}
\par}}
\smallskip
\noindent 
and closed under the following inference rules:
		{\small{
		\begin{gather*}
			\frac{\phi\vdash \chi\quad \chi\vdash \psi}{\phi\vdash \psi}
			\ \ 
			\frac{\phi\vdash \psi}{\phi\left(\chi/p\right)\vdash\psi\left(\chi/p\right)}
			\ \ 
			\frac{\chi\vdash\phi\quad \chi\vdash\psi}{\chi\vdash \phi\wedge\psi}
			\ \ 
			\frac{\phi\vdash\chi\quad \psi\vdash\chi}{\phi\vee\psi\vdash\chi}
\ \ 
			\frac{\phi\vdash\psi}{\Box \phi\vdash \Box \psi}
\ \ 
\frac{\phi\vdash\psi}{\Diamond \phi\vdash \Diamond \psi}
\end{gather*}}}

\noindent In the following part, we define the polarity-based semantics for this logic.

\fakeparagraph{Relational semantics.}
\label{sssec:relsem} 
The following preliminaries are taken from \cite{conradie2021rough,conradie2020non}.
For any binary relation $T\subseteq U\times V$, and any $U'\subseteq U$  and $V'\subseteq V$,  we let

{{
\centering
$T^{(1)}[U']\ceq\{v\mid \forall u(u\in U'\Rightarrow uTv) \}  \quad\quad T^{(0)}[V']\ceq\{u\mid \forall v(v\in V'\Rightarrow uTv)\}. $
\par
}}

\noindent For any $u \in U$ (resp.~$v \in V$) we will write $T^{(1)}[u]$ (resp.~$T^{(0)}[v]$) in place of $T^{(1)}[\{u\}]$ (resp.~$T^{(0)}[\{v\}]$).

A {\em polarity} or {\em formal context} (cf.~\cite{ganter2012formal}) is a tuple $\mathbb{P} =(A,X,I)$, where $A$ and $X$ are sets, and $I \subseteq A \times X$ is a binary relation. 
$A$ and $X$ can be understood as the collections of {\em objects} and {\em features}, and for any $a\in A$ and $x\in X$, $aIx$ exactly when the object $a$ has the feature $x$. For any polarity $\mathbb{P} = (A, X, I)$, the pair of maps 

{{\centering
 $(\cdot)^\uparrow: \mathcal{P}(A)\to \mathcal{P}(X)$ and $(\cdot)^\downarrow: \mathcal{P}(X)\to \mathcal{P}(A),$
\par}}
\smallskip 

\noindent defined by $B^\uparrow \ceq I^{(1)}[B]$ and $Y^\downarrow \ceq I^{(0)}[Y]$ where $B\subseteq A$ and $Y\subseteq X$, forms a Galois connection, and hence induces the closure operators$(\cdot)^{\uparrow\downarrow}$ and $(\cdot)^{\downarrow\uparrow}$ on $\mathcal{P}(A)$ and on $\mathcal{P}(X)$, respectively. Again, we will write $a^\uparrow$ and $a^{\uparrow\downarrow}$ (resp.~$x^\downarrow$ and $x^{\downarrow\uparrow}$) in place of $\{a\}^\uparrow$ and $\{a\}^{\uparrow\downarrow}$ (resp.~$\{x\}^\downarrow$ and $\{x\}^{\downarrow\uparrow}$).


A {\em formal concept} of a polarity $\mathbb{P}=(A,X,I)$ is a tuple $c=(\val{c},\descr{c})$ such that $\val{c}\subseteq A$ and $\descr{c}\subseteq X$, and $\val{c} = \descr{c}^\downarrow$ and $\descr{c} = \val{c}^\uparrow$, i.e. the sets $\val{c}$ and $\descr{c}$ are  Galois-stable. The set of formal concepts of polarity $\mathbb{P}$, with the order defined by

{{\centering
$c_1 \leq c_2 \quad \text{iff} \quad \val{c_1} \subseteq \val{c_2} \quad \text{iff} \quad \descr{c_2} \subseteq \descr{c_1}$,
\par
}}

\noindent forms a complete lattice $\mathbb{P}^+$, namely the {\em concept lattice} of $\mathbb{P}$. 

An {\em enriched formal context} is a tuple $\mathbb{F} =(\mathbb{P}, R_\Box, R_\Diamond)$, where $R_\Box  \subseteq A \times X$ and $ R_\Diamond \subseteq X \times A$ are {\em $I$-compatible} relations, that is, for all $a \in A$ and $x \in X$, the sets $R_\Box^{(0)}[x]$, $R_\Box^{(1)}[a]$, $R_\Diamond^{(0)}[a]$, $R_\Diamond^{(1)}[x]$ are Galois-stable in $\mathbb{P}$. 
Given the operations $[R_\Box]$ and $\langle R_\Diamond\rangle$ on $\mathbb{P}^+$ corresponding to $R_\Box$ and $R_\Diamond$, respectively, we have for any $c \in \mathbb{P}^+$,

{{
\centering
      $[R_\Box] c =(R_\Box^{(0)}[\descr{c}], I^{(1)}[R_\Box^{(0)}[\descr{c}]]) \quad \text{and} \quad \langle R_\Diamond\rangle  c =( I^{(0)}[R_\Diamond^{(0)}[\val{c}]], R_\Diamond^{(0)}[\val{c}]).$
\par
}}

\noindent We refer to the algebra $\mathbb{F}^+=(\mathbb{P}^+, [R_\Box], \langle R_\Diamond \rangle)$ as the {\em complex algebra} of  $\mathbb{F}$. A {\em valuation} on such an $\mathbb{F}$
is a map $V\colon\Prop\to \mathbb{P}^+$. For each $p\in \Prop$, we let $\val{p} \ceq \val{V(p)}$ (resp.~$\descr{p}\ceq \descr{V(p)}$) denote the extension (resp. intension) of the interpretation of $p$ under $V$.

A {\em model} is a tuple $\mathbb{M} = (\mathbb{F}, V)$, where $\mathbb{F} = (\mathbb{P}, R_{\Box}, R_{\Diamond})$ is an enriched formal context and $V$ is a valuation of $\mathbb{F}$. For every $\phi\in \mathcal{L}$, we let $\val{\phi}_\mathbb{M} \ceq \val{V(\phi)}$ (resp.~$\descr{\phi}_\mathbb{M}\ceq \descr{V(\phi)}$) denote the extension (resp. intension) of the interpretation of $\phi$ under the homomorphic extension of $V$.   The `satisfaction' and `co-satisfaction' relations $\Vdash$ and $\succ$  can be recursively defined as follows: 
\smallskip

{{\centering 
\begin{tabular}{l@{\hspace{1em}}l@{\hspace{2em}}l@{\hspace{1em}}l}
$\mathbb{M}, a \Vdash p$ & iff $a\in \val{p}_{\mathbb{M}}$ &
$\mathbb{M}, x \succ p$ & iff $x\in \descr{p}_{\mathbb{M}}$ \\
$\mathbb{M}, a \Vdash\top$ & always &
$\mathbb{M}, x \succ \top$ & iff   $a I x$ for all $a\in A$\\
$\mathbb{M}, x \succ  \bot$ & always &
$\mathbb{M}, a \Vdash \bot $ & iff $a I x$ for all $x\in X$\\
$\mathbb{M}, a \Vdash \phi\wedge \psi$ & iff $\mathbb{M}, a \Vdash \phi$ and $\mathbb{M}, a \Vdash  \psi$ & 
$\mathbb{M}, x \succ \phi\wedge \psi$ & iff $(\forall a\in A)$ $(\mathbb{M}, a \Vdash \phi\wedge \psi \Rightarrow a I x)$
\\
$\mathbb{M}, x \succ \phi\vee \psi$ & iff  $\mathbb{M}, x \succ \phi$ and $\mathbb{M}, x \succ  \psi$ & $\mathbb{M}, a \Vdash \phi\vee \psi$ & iff $(\forall x\in X)$ $(\mathbb{M}, x \succ \phi\vee \psi \Rightarrow a I x)$.
\end{tabular}
\par}}
\smallskip

\noindent As to the interpretation of modal operators: 
\smallskip

{{\centering
\begin{tabular}{llcll}
$\mathbb{M}, a \Vdash \Box\phi$ &  iff $(\forall x\in X)(\mathbb{M}, x \succ \phi \Rightarrow a R_\Box x)$ &&
$\mathbb{M}, x \succ \Box\phi$ &  iff $(\forall a\in A)(\mathbb{M}, a \Vdash \Box\phi \Rightarrow a I x)$\\

$\mathbb{M}, x \succ \Diamond\phi$ &  iff $(\forall a\in A) (\mathbb{M}, a \Vdash \phi \Rightarrow x R_\Diamond a)$ &&
$\mathbb{M}, a \Vdash \Diamond\phi$ & iff $(\forall x\in X)(\mathbb{M}, x \succ \Diamond\phi \Rightarrow a I x)$.  \\

\end{tabular}
\par}}
\smallskip

\noindent The definition above ensures that, for any $\mathcal{L}$-formula $\varphi$,

{{
\centering
$\mathbb{M}, a \Vdash \phi$ \quad iff \quad $a\in \val{\phi}_{\mathbb{M}}$, \qquad  and \qquad$\mathbb{M},x \succ \phi$ \quad iff \quad $x\in \descr{\phi}_{\mathbb{M}}$. \par}}

{{
\centering
$\mathbb{M}\models \phi\vdash \psi$ \quad iff \quad $\val{\phi}_{\mathbb{M}}\subseteq \val{\psi}_{\mathbb{M}}$\quad  iff  \quad  $\descr{\psi}_{\mathbb{M}}\subseteq \descr{\phi}_{\mathbb{M}}$. 
\par}}

The interpretation of the propositional connectives $\vee$ and $\wedge$ in the framework described above reproduces the standard notion of join and meet of formal concepts used in FCA. The interpretation of operators $\Box$ and $\Diamond$ is motivated by algebraic properties and duality theory for modal operators on lattices (see~\cite[Section 3]{conradie2020non} for an expanded discussion). 
In \cite[Proposition 3.7]{conradie2021rough}, it is shown that the semantics of LE-logics is compatible with Kripke semantics for classical modal logic, and thus, LE-logics are indeed generalizations of classical modal logic.
This interpretation is further justified in \cite[Section 4]{conradie2021rough} by noticing  that, under 
the interpretations of the relation $I$ as 
$a I x$ iff ``object $a $  has feature $x$''
and  $R=R_\Box =R^{-1}_\Diamond$ as $a R x$ iff ``there is evidence that object $a$  has feature $x$'', then, for any concept $c$, the extents of concepts $\Box c$ and $\Diamond c$ can be interpreted as ``the set of objects which {\em certainly} belong to $c$'' (upper approximation), and ``the set of objects which {\em possibly} belong to $c$'' (lower approximation) respectively. Thus, the interpretations of $\Box$ and $\Diamond$ have similar meaning in the LE-logic as in the classical modal logic.

\subsection{Description logic LE-$\mathcal{ALC}$}\label{Ssec: LE-DL}
In this section, we recall the lattice-based description logic LE-$\mathcal{ALC}$ introduced in \cite{van2023non} as a counterpart of non-distributive modal logic. It serves as a natural framework in the realm of description logics for reasoning about the (enriched) formal contexts and the concepts defined by them.

The language of LE-$\mathcal{ALC}$ contains two types of individuals, usually interpreted as {\em objects} and {\em features}. Let $\mathsf{OBJ}$ and $\mathsf{FEAT}$ be disjoint sets of individual names for objects and features. The set $\mathcal{R}$ of the role names for LE-$\mathcal{ALC}$ is the union of three types of relations: (1) a unique relation $I \subseteq \mathsf{OBJ} \times \mathsf{FEAT}$; (2) a set of relations $\mathcal{R}_\Box$  of the form $R_\Box \subseteq \mathsf{OBJ} \times \mathsf{FEAT}$; (3) a set of relations $\mathcal{R}_\Diamond$ of the form $ R_\Diamond \subseteq \mathsf{FEAT} \times \mathsf{OBJ}$. The relation $I$ is intended to be interpreted as the incidence relation of formal contexts and encodes information on which objects have which features, and the relations in $\mathcal{R}_\Box$  and $\mathcal{R}_\Diamond$ encode additional relationships between objects and  features (see~\cite{conradie2021rough} for an extended discussion). In this paper, we work with an LE-$\mathcal{ALC}$ language in which the sets of role names $\mathcal{R}_\Box$ and $\mathcal{R}_\Diamond$ are singletons. All the results in this paper can be generalized to language with multiple role names in each of these sets straightforwardly.

For any set $\mathcal{D}$ of atomic concept names, the language of LE-$\mathcal{ALC}$ concepts is:

{{\centering
 $C \ceq D\ |\ C_1 \wedge C_2\ |\ C_1\vee C_2\ |\ \langle R_\Diamond \rangle C\ |\ [R_\Box ]C$   
\par}}

\noindent where $D \in \mathcal{D}$.
This language matches the LE-logic language and has an analogous intended interpretation of the complex algebras of the enriched formal contexts (cf.~Section \ref{ssec:LE-logic}).
As usual in FCA, $\vee$ and $\wedge$ are to be interpreted as the smallest common superconcept and the greatest common subconcept. 
We do not use the symbols $\forall r$ and $\exists r$ in the context of LE-$\mathcal{ALC}$ because using the same notation verbatim would be ambiguous or misleading, as the semantic clauses of modal operators in LE-logic use the universal quantifiers. 

TBox assertions in LE-$\mathcal{ALC}$ are of the shape $C_1 \equiv C_2$, where $C_1$ and $C_2$ are concepts defined as above. As is standard in DL (see~\cite{DLhandbook} for more details), general concept inclusions of the form $C_1 \sqsubseteq C_2$ can be rewritten as $C_1 \equiv C_2 \wedge C_3$, where $C_3$ is a new concept name.
ABox assertions are of the form:
\noindent

{{\centering
  $aR_\Box x,\quad xR_\Diamond a,\quad aIx,\quad a:C,\quad x::C,\quad \neg \alpha,$  
\par}}

\noindent
where $\alpha$ is any of the first five ABox terms. We refer to the first three types of terms as {\em relational terms}. We denote an arbitrary ABox (resp.~TBox) with $\mathcal{A}$ (resp.~$\mathcal{T}$). The interpretations of the terms $a:C$ and $x::C$ are: ``object $a$ is a member of concept $C$'', and  ``feature  $x$ is in the description of concept $C$'', respectively. Note that we explicitly add negative terms to ABoxes, as the concept names in LE-$\mathcal{ALC}$ do not contain negations. 

An {\em interpretation} for LE-$\mathcal{ALC}$ is a tuple $\mathcal{M} = (\mathbb{F}, \cdot^\mathcal{M}) $, where $\mathbb{F}=(\mathbb{P}, R_\Box, R_\Diamond)$ is
an enriched formal context, and $\cdot^\mathcal{M}$ maps:

\noindent 1. individual names $a \in \mathsf{OBJ}$ (resp.~$x \in \mathsf{FEAT}$)   to some $a^\mathcal{M} \in A$ (resp.~$x^\mathcal{M}\in X$);

\noindent 2. role names $I$, $R_\Box$ and $R_\Diamond$ to  relations $I^\mathcal{M}\subseteq A\times X$, $R_\Box^\mathcal{M}\subseteq A\times X$ and $R_\Diamond^\mathcal{M}\subseteq X\times A$ in $\mathbb{F}$;

\noindent 3. any atomic concept $D$ to $D^\mathcal{M}\in \mathbb{F}^+$, and other concepts as follows:

{{\centering
    \begin{tabular}{l l ll}
    $(C_1 \wedge C_2)^{\mathcal{M}} = C_1^{\mathcal{M}} \wedge C_2^{\mathcal{M}}$   & $(C_1\vee C_2)^{\mathcal{M}} = C_1^{\mathcal{M}} \vee C_2^{\mathcal{M}}$ &
    $([R_\Box]C)^\mathcal{M} = [R_\Box^\mathcal{M}]C^{\mathcal{M}}$  &  $(\langle R_\Diamond \rangle C)^\mathcal{M} =\langle  R_\Diamond^{\mathcal{M}} \rangle C^{\mathcal{M}} $\\
    \end{tabular}
\par}}

\noindent where all the connectives are interpreted as defined in LE-logic (cf.~Section \ref{ssec:LE-logic}). The satisfiability relation for an interpretation $\mathcal{M}$ is defined as follows:

\noindent 1. $\mathcal{M}\models C_1\equiv C_2$ iff $\val{C_1^\mathcal{M}} = \val{C_2^\mathcal{M}}$ iff $\descr{C_2^\mathcal{M}} = \descr{C_1^\mathcal{M}}$.

\noindent 2. $\mathcal{M} \models a:C$ iff $a^\mathcal{M} \in \val{C^\mathcal{M}}$ and $\mathcal{M} \models x::C$ iff $x^\mathcal{M} \in \descr{C^\mathcal{M}}$.

\noindent 3. $ \mathcal{M} \models a I x$ (resp.~$a R_\Box x$, $x R_\Diamond a$) iff $a^{\mathcal{M}} I^{\mathcal{M}} x^{\mathcal{M}} $ (resp.~$a^{\mathcal{M}} R_\Box^{\mathcal{M}} x^{\mathcal{M}} $, $x^{\mathcal{M}} R_\Diamond^{\mathcal{M}} a^{\mathcal{M}} $). 

\noindent 4. $ \mathcal{M} \models  \neg \alpha$, where $\alpha$ is any ABox term, iff $ \mathcal{M} \not\models  \alpha$. 

The satisfaction definition can be extended to concept inclusion as follows. For any concepts $C_1$, and $C_2$, and an interpretation $\mathcal{M}$,  $\mathcal{M} \models C_1 \sqsubseteq C_2$ iff $C_1^{\mathcal{M}}  \leq C_2^{\mathcal{M}}$. 

An interpretation $\mathcal{M}$ is a {\em model} for an LE-$\mathcal{ALC}$ knowledge base $(\mathcal{A}, \mathcal{T})$, where $\mathcal{A}$ is an ABox and $\mathcal{T}$ is a TBox, if $\mathcal{M}\models \mathcal{A}$ and $\mathcal{M} \models \mathcal{T}$. An LE-$\mathcal{ALC}$ knowledge base $(\mathcal{A}, \mathcal{T})$ is said to be inconsistent if there is no model for it. We say an ABox $\mathcal{A}$ is consistent if knowledge base $(\mathcal{A}, \mathcal{T})$ has a model and $\mathcal{T}$ is empty. In this paper, we use LE-$\mathcal{ALC}$ knowledge bases to mean LE-$\mathcal{ALC}$ knowledge bases with acyclic Tboxes unless otherwise stated.

\subsection{Tableaux algorithm for checking LE-$\mathcal{ALC}$ ABox consistency}
\label{Ssec: tableau}
In this section, we introduce the tableaux algorithm for checking the consistency of LE-$\mathcal{ALC}$ ABoxes. An LE-$\mathcal{ALC}$ ABox $\mathcal{A}$ contains a {\em clash} iff it contains both $\beta$ and $\neg \beta$ for some relational term $\beta$. The expansion rules below are designed so that the expansion of $\mathcal{A}$ will contain a clash iff  $\mathcal{A}$ is inconsistent. 
The set $sub(C)$ of sub-formulas of any LE-$\mathcal{ALC}$ concept name $C$ is defined as usual. A concept name $C'$ {\em occurs} in the ABox $\mathcal{A}$ (denoted as $C' \in \mathcal{A}$) if $C'\in sub(C)$ for some $C$ such that one of the terms $a:C$, $x::C$, $\neg (a:C)$, or $\neg (x ::C)$ is in $\mathcal{A}$. A constant $b$ (resp.~$y$) {\em occurs} in $\mathcal{A}$ ($b \in \mathcal{A}$, or $y \in \mathcal{A}$), iff some term containing $b$ (resp.~$y$) occurs in it. 

The tableaux algorithm below provides a method to construct a model $(\mathbb{F},\cdot^\mathcal{M})$ for every consistent $\mathcal{A}$, where $\mathbb{F}= (\mathbb{P}, R_\Box, R_\Diamond)$ is such that, for any $C \in \mathcal{A}$, some $a_C \in A$ and $x_C \in X$ exist such that, for any $a \in A$ (resp.~any $x \in X$), $a \in \val{C^{\mathcal{M}}}$ (resp.~$x \in \descr {C^{\mathcal{M}}}$) iff $a I x_C$ (resp.~$a_C I x$).
We call $a_C$ and $x_C$ the {\em classifying object} and the {\em classifying feature} of $C$, respectively. To make the notation easily readable, we write $a_{\Box C}$, $x_{\Box C}$ (resp.~$a_{\Diamond C}$, $x_{\Diamond C}$) instead of $a_{[R_\Box]C}$, $x_{[R_\Box]C}$ (resp.~$a_{\langle R_\Diamond\rangle C}$, $x_{\langle R_\Diamond\rangle C}$). The commas in each rule are meta-linguistic conjunctions, hence every tableau is non-branching. 

\begin{algorithm} 
\caption{tableaux algorithm for checking LE-$\mathcal{ALC}$ ABox consistency }\label{alg:main algo}
    \hspace*{\algorithmicindent} \textbf{Input}: An   LE-$\mathcal{ALC}$ ABox $\mathcal{A}$. \quad \textbf{Output}: whether $\mathcal{A}$ is inconsistent. 
    \begin{algorithmic}[1]
        \State \textbf{if} there is a clash in $\mathcal{A}$ \textbf{then} \textbf{return} ``inconsistent''.
         \State \textbf{pick} any applicable expansion rule $R$, \textbf{apply} $R$ to $\mathcal{A}$ and proceed recursively.  
          \State \textbf{if} there is no clash after  post-processing  \textbf{return} ``consistent".
    \end{algorithmic}
\end{algorithm}

{{
\centering
\begin{tabular}{cc} 
\mc{1}{c}{\textbf{Creation rule}} & \mc{1}{c}{\textbf{Basic rule}} \\
\AXC{For any $C \in \mathcal{A}$}
\LL{\fns create}
\UIC{$a_C:C$, \quad $x_C::C$}
\DP
 \ & \ 
\rule[-1.85mm]{0mm}{8mm}
\AXC{$b:C, \quad y::C$}
\LL{\fns $I$}
\UIC{$b I y$}
\DP 
\end{tabular}

\begin{tabular}{cccc}
\mc{2}{c}{\textbf{Rules for the logical connectives}} & \mc{2}{c}{\textbf{$I$-compatibility rules}} \\
\rule[-1.85mm]{0mm}{8mm}
 \AXC{$b:C_1 \wedge  C_2$}
\RL{\fns $\wedge_A$}
\UIC{$b:C_1,$ \quad $b:C_2$}
\DP 
\ & \

\AXC{$y::C_1 \vee  C_2$}
\LL{\fns $\vee_X$}
\UIC{$y::C_1,$ \quad $y::C_2$}
\DP 
\  &  \

\AXC{$b I \Box y$}
\LL{\fns $\Box y$}
\UIC {$b R_\Box y$}
\DP

\ & \

\AXC{$b I \blacksquare y$}
\RL{\fns $\blacksquare y$}
\UIC {$y R_\Diamond b$}
\DP
\\[3mm]

 \AXC{$b:[R_\Box]C,$ \quad  $y::C$}
\LL{\fns $\Box$}
\UIC{$b R_\Box y$}
\DP 
\ & \ 
  \AXC{$y::\langle  R_\Diamond \rangle C,$ \quad  $b:C$}
\RL{\fns $\Diamond$}
\UIC{$y R_\Diamond b$}
\DP 
\ & \
\AXC{$\Diamond b I  y$}
\LL{\fns $\Diamond b$}
\UIC {$y R_\Diamond b$}
\DP
\ & \
\AXC{$\Diamondblack b I y$}
\RL{\fns $\Diamondblack b$}
\UIC {$b R_\Box y$}
\DP
\\[3 mm]

\end{tabular}
\begin{tabular}{rl}
\mc{2}{c}{\textbf{inverse rules for connectives}}  \\

   \AXC{$b:C_1$, $b:C_2$, $C_1 \wedge C_2 \in \mathcal{A}$}
\LL{\fns $\wedge_A^{-1}$}
\UIC {$b:C_1 \wedge C_2$}
\DP
\ & \
\AXC{$y::C_1$, $y::C_2$, $C_1 \vee C_2 \in \mathcal{A}$}
\RL{\fns $\vee_X^{-1}$}
\UIC {$y::C_1 \vee C_2$}
\DP
\end{tabular}
\begin{tabular}{rl}
\mc{2}{c}{\textbf{Adjunction rules}}  \\
\rule[-1.85mm]{0mm}{8mm}
\AXC{$ b R_\Box y$}
\LL{\fns $R_\Box$}
\UIC{$\Diamondblack b I y,$ \quad  $b I \Box y$}
\DP 
\ & \ 
\AXC{$y R_\Diamond b$}
\RL{\fns $R_\Diamond$}
\UIC{$\Diamond b I y,$ \quad  $b I \blacksquare y$}
\DP 
\\ [2mm]
\end{tabular}

\begin{tabular}{cccc}
\mc{2}{c}{\textbf{Basic rules for negative assertions}} & \mc{2}{c}{\textbf{Appending  rules}} \\
\rule[-1.85mm]{0mm}{8mm}
\AXC{$\neg (b:C)$}
\LL{\fns $\neg b$}
\UIC{$\neg (b I x_C)$}
\DP 
\ & \ 
\AXC{$\neg (x::C)$}
\RL{\fns $\neg x$}
\UIC{$\neg (a_C I x)$}
\DP 
\ & \ 
\AXC{$b I x_C$}
\LL{\fns $x_C$}
\UIC{$b:C$}
\DP 
\ & \ 
\AXC{$a_C I y$}
\RL{\fns $a_C$}
\UIC{$y::C$}
\DP 
\\
\end{tabular}
\par}}
\smallskip

\noindent Note that in the creation rule, the   $a_C$ and $x_C$ are new (i.e.~different from any names already appearing in the tableaux) special object and feature names unique for each $C$.  In the adjunction rules, the individuals $\Diamondblack b$, $\Box y$, $\Diamond b$, and $\blacksquare y$ are new and unique individual names\footnote{The new individual names $\Diamondblack b$, $\Diamond b$, $\Box y$, and $\blacksquare y$ appearing in tableaux expansion are purely syntactic entities. Intuitively, they correspond to the classifying objects (resp.~features) of the concepts $\Diamondblack \textbf{b}$, $\Diamond \textbf{b}$ (resp.~$\Box \textbf{y}$, resp.~$\blacksquare \textbf{y}$), where $\mathbf{b}=(b^{\uparrow\downarrow}, b^\uparrow)$ (resp.~$\mathbf{y}=(y^\downarrow, y^{\downarrow\uparrow})$) is the concept generated by $b$ (resp.~$y$), and the operation  $\Diamondblack$ (resp.$\blacksquare$) is the left (resp.~right) adjoint of operation $\Box$ (resp.~$\Diamond$).} for relations $R_\Box$ and $R_\Diamond$, and individuals $b$ and $y$, except for $\Diamond a_C= a_{\Diamond C}$ and $\Box x_C= x_{\Box C}$. Side conditions that the conjunction and disjunction occur in $\mathcal{A}$ for rules $\wedge_A^{-1}$ and $\vee_X^{-1}$ ensure that we do not add new meets or joins to the concept names.

The following theorem follows from the results in \cite{van2023non}:


For any consistent LE-$\mathcal{ALC}$ ABox $\mathcal{A}$, the tableaux completion $\overline{\mathcal{A}}$ of $\mathcal{A}$ is a set of assertions which are obtained by applying the tableaux algorithm~\ref{alg:main algo} to $\mathcal{A}$. From $\overline{\mathcal{A}}$, we can construct a model $\mathcal{M} =(\mathbb{F},\cdot^\mathcal{M})$, where $\mathbb{F} =(A,X,I,R_\Box,R_\Diamond)$ is described as follows:
$A$ and $X$ are taken to be the sets of all individual names of object and feature that occur in $\overline{\mathcal{A}}$, respectively, and all individuals are interpreted by their names. For any role name $R$, its interpretation $ R^{\mathcal{M}} $ is defined as follows: for any individual names $l$, $m$, $l R^{\mathcal{M}} m$ iff $ l R m \in \overline{\mathcal{A}}$. Finally, for the atomic concept $D$, its interpretation is set to the concept $(x_D^\downarrow, a_D^\uparrow)$. The following result was proved in \cite{van2023non}.

\begin{theorem}\label{th:model}
For any LE-$\mathcal{ALC}$ ABox $\mathcal{A}$, 
\begin{itemize}
    \item the tableaux algorithm applied to $\mathcal{A}$ terminates in polynomial time in $|\mathcal{A}|$;
    \item $\overline{\mathcal{A}}$ contains a clash iff $\mathcal{A}$ is inconsistent;
    \item if $\mathcal{A}$ is inconsistent, then the model $\mathcal{M}$ as constructed above is a model for $\mathcal{A}$ of the size polynomial in $|\mathcal{A}|$. Moreover, for any individual names $b$, $y$, and concept $C$ occurring in $\mathcal{A}$, $b \in \val{C}$ iff $b I x_C \in \overline{\mathcal{A}}$, and $y \in \descr{C}$ iff $a_c I y \in  \overline{\mathcal{A}}$.
\end{itemize}
 
\end{theorem}

\begin{remark}
    The algorithm can be easily extended to acyclic TBoxes (exponential-time), using the unraveling technique (see~\cite{DLbook} for details).
\end{remark}

\subsection{Ontology-mediated query answering}\label{ssec:Ontology mediated query answering}
A key task in description logic ontologies (knowledge bases) is to support various reasoning tasks, one of which is to answer queries based on ontologies~\cite[chapter 7]{DLhandbook}.
Let $\mathcal{K} =(\mathcal{A}, \mathcal{T})$ be a consistent knowledge base given in a specific description logic $\mathrm{DL}$. Given a query $q(\overline{p})$ (with a possibly empty tuple of free variables $\overline{p}$) in (appropriate) first-order language and a model $\mathcal{M}$ of $\mathcal{K}$, we say that a sequence of individuals $\overline{a}$ in $\mathcal{A}$ is an {\em answer for query $q(\overline{p})$} w.r.t.~model $\mathcal{M}$ of knowledge base $\mathcal{K}$ if $\mathcal{M} \models q(\overline{a})$. An answer $\overline{a}$ in $\mathcal{A}$ is said to be {\em a certain answer for the query $q(\overline{p})$} with respect to a knowledge base $\mathcal{K}$ if it is an answer for $q(\overline{p})$ w.r.t.~all the models of $\mathcal{K}$. An important notion used in ontology-mediated query answering is that of a {\em universal} or {\em canonical} model. For a query $q(\overline{p})$ on a knowledge base $\mathcal{K}$, we say that a model $\mathcal{M}$ of $\mathcal{K}$ is a universal or canonical model for $\mathcal{K}$ if for any $\overline{a}$ appearing in $\mathcal{K}$, $\mathcal{K} \models q(\overline{a})$ iff $\mathcal{M} \models q(\overline{a})$. In case $\overline{p}$ is empty, the answer or the certain answer for such query is true or false depending on whether $\mathcal{M} \models q$ or not. Thus, we can provide certain answer for query $q$ over $\mathcal{K}$ by only looking over the universal or canonical model $\mathcal{M}$. Universal models for different description logics have been extensively studied~\cite{calvanese2007tractable,glimm2008conjunctive,Kontchakov2014,bienvenu2015ontology}. In this paper, we would focus on answering some specific types of queries over knowledge bases in non-distributive description logic LE-$\mathcal{ALC}$. To this end, we show that for any LE-$\mathcal{ALC}$ ABox $\mathcal{A}$, the model constructed from it by applying the LE-$\mathcal{ALC}$ tableaux algorithm \ref{alg:main algo} acts as the universal model for $\mathcal{A}$ w.r.t.~several different types of queries. As the tableaux algorithm is polynomial in time and produces a polynomial size model in $|\overline{\mathcal{A}}|$, this provides a polynomial-time algorithm to answer these types of queries.


\section{Query answering over LE-$\mathcal{ALC}$ ABoxes}\label{sec:query answering}
In this section, we discuss different types of queries pertaining to  LE-$\mathcal{ALC}$ ABoxes and develop algorithms to answer them. We start by showing that for any consistent LE-$\mathcal{ALC}$ ABox $\mathcal{A}$, the model obtained  using Algorithm \ref{alg:main algo} behaves like universal model w.r.t.~several types of queries. 

\subsection{Universal model for LE-$\mathcal{ALC}$ ABox}
For any individual name appearing in the tableaux expansion $\overline{\mathcal{A}}$ of an LE-$\mathcal{ALC}$ ABox $\mathcal{A}$, we define its {\em concept companion} as follows :

1.~For any constant $b$ (resp.~$y$) appearing in $\mathcal{A}$,  $con(b)$ (resp.~$con(y)$) is a concept such that for any interpretation $\mathcal{M}$, $con(b)^\mathcal{M} = \mathbf{b}$ (resp.~$con(y)^\mathcal{M} = \mathbf{y}$), where $\mathbf{b}$ (resp.~$\mathbf{y}$) denotes the concept generated by $b^\mathcal{M}$ (resp.~$y^\mathcal{M}$), i.e. $\mathbf{b}=(({b^\mathcal{M}})^{\uparrow\downarrow}, ({b^\mathcal{M}})^\uparrow)$ (resp.~$\mathbf{y}=(({y^\mathcal{M}})^\downarrow, ({y^\mathcal{M}})^{\downarrow\uparrow})$). 

2.~For any constant $\Diamond b$ (resp.~$\Diamondblack b$, resp.~$\blacksquare y$, resp.~$\Box y$) appearing in $\mathcal{A}$,  $con(\Diamond b)= \Diamond con(b)$ (resp. 
$con(\Diamondblack b)= \Diamondblack con(b)$, resp.~$con(\blacksquare y)= \blacksquare con(y)$, resp.~$con(\Box y)= \Box con(y)$), where the operation $\Diamondblack$ (resp. 
$\blacksquare$) is the left (resp.~right) adjoint of $\Box$ (resp.~$\Diamond$). 
 
\begin{lemma}\label{lem:universal model properties}
For any consistent LE-$\mathcal{ALC}$ ABox $\mathcal{A}$,  individual names $b$, $y$ appearing in its completion $\overline{\mathcal{A}}$, and concept $C$ appearing in $\mathcal{A}$:

\smallskip
{{
\centering
\begin{tabular}{ll}
   1.~$\mathcal{A} \models  con(b) \sqsubseteq   con(y)$  iff $b I y \in \overline{\mathcal{A}}$,  & 2.~$\mathcal{A} \models  con(b) \sqsubseteq  \Box con (y)$  iff $b R_\Box  y \in \overline{\mathcal{A}}$, \\
   3.~$\mathcal{A} \models  \Diamond con(b) \sqsubseteq   con (y)$  iff $y R_\Diamond  b \in \overline{\mathcal{A}}$,  & 4.~$\mathcal{A} \models  con(b) \sqsubseteq   C$  iff $b I x_C \in \overline{\mathcal{A}}$, \\
   5.~$\mathcal{A} \models  C \sqsubseteq  con(y)  $  iff $a_C I y \in \overline{\mathcal{A}}$, & 6.~$\mathcal{A} \models  con(b) \sqsubseteq  C  $  iff $b:C \in \overline{\mathcal{A}}$, \\
   7.~$\mathcal{A} \models  C \sqsubseteq  con(y)$  iff $y::C \in \overline{\mathcal{A}}$, & 8.~$\mathcal{A} \models  C_1 \sqsubseteq  C_2$  iff $a_{C_1} I x_{C_2}\in \overline{\mathcal{A}}$. \\
\end{tabular} 
\par
}}
\end{lemma}

\begin{proof}
The proofs from left to right for items 1-5 follow immediately from Theorem~\ref{th:model}. We prove the right to left implications by simultaneous (over all the items) induction on the number of expansion rules applied. The base case is when the term in the right appears in $\mathcal{A}$. In this case, it is immediate from the definition that we get the required condition on the left. 

    \textbf{Creation rule.} By this rule, $a_C : C$ and $x_C :: C$ are added by any $C\in\mathcal{A}$, which imply $C\sqsubseteq C$.

    \textbf{Basic rule.} By this rule, $bIy$ is added from $b:C$ and $y::C$. By induction applied to items 6 and 7, we get $con(b)\sqsubseteq C$ and $C\sqsubseteq con(y)$, which imply that $con(b)\sqsubseteq con(y)$. It is easy to check item 4 and item 5 also hold. For item 8, $a_{C_1} I x_{C_2}$ is added from $a_{C_1}:C$ and $x_{C_2}::C$. By induction applied to items 6 and 7, we have $C_1\sqsubseteq C$ and $C\sqsubseteq C_2$, which imply that $C_1\sqsubseteq C_2$.

    \textbf{Rules $\wedge_A$,  $\vee_X$, $\wedge_A^{-1}$, $\vee_X^{-1}$.} We give the proofs for rules $\wedge_A$ and $\wedge_A^{-1}$. The proofs for $\vee_X$ and $\vee_X^{-1}$ are analogous. 
    By rule $\wedge_A$, $b:C_1$ and $b:C_2$ are added from $b:C_1\wedge C_2$. By induction applied to item 6, $con(b)\sqsubseteq C_1\wedge C_2$, and thus $con(b)\sqsubseteq C_1$ and $con(b)\sqsubseteq C_2$. 
    By rule $\wedge^{-1}_A$, $b:C_1\wedge C_2$ is added from $b:C_1$, $b:C_2$, and $C_1\wedge C_2\in \mathcal{A}$. By induction applied to item 6, we have $con(b)\sqsubseteq C_1$ and $con(b)\sqsubseteq C_2$. Since $C_1\wedge C_2$ exists in $\mathcal{A}$, we get $con(b)\sqsubseteq C_1\wedge C_2$.
    
    \textbf{I-compatibility rules.} We give the proofs for rules $\Box y$ and $\blacksquare y$. The proofs for $\Diamond b$ and $\Diamondblack b$ are analogous. 
    By rule $\Box y$, $b R_\Box y$ is added from $b I \Box y$, by induction applied to item 1, we get $con(b)\sqsubseteq con(\Box y)$ and by definition we have $con(b)\sqsubseteq \Box con(y)$. 
    By rule $\blacksquare y$, $y R_\Diamond b$ is added from $b I \blacksquare y$. By induction applied to item 1, we have $con(b)\sqsubseteq con(\blacksquare y)$, and by definition $con(b)\sqsubseteq\blacksquare con(y)$. By adjunction, we have $\Diamond con(b)\sqsubseteq con(y)$.
    
    \textbf{Rules $\Box$ and $\Diamond$.} We give the proof for rule $\Box$, and the proof for $\Diamond$ is analogous. By rule $\Box$, $b R_\Box y$ is added from $b:[R_\Box]C$ and $y::C$. By induction applied to item 6, we have $con(b)\sqsubseteq\Box C$. By induction on item 7, we get $C\sqsubseteq con(y)$. As $\Box$ is a monotone operator, we have $\Box C\sqsubseteq\Box con(y)$. Thus, $con(b)\sqsubseteq\Box con(y)$. 

    \textbf{Adjunction rules.} We give the proof for rule $R_\Box$ and the proof for $R_\Diamond$ is analogous. By rule $R_\Box$, $b I \Box y$ (resp.~$\Diamondblack b I y$) is added from $b R_\Box y$. By induction applied to item 2, we have $con(b)\sqsubseteq\Box con(y)$ (resp.~$\Diamondblack con(b)\sqsubseteq con(y)$ by adjunction), and thus $con(b)\sqsubseteq con(\Box y)$ (resp.~$con(\Diamondblack b)\sqsubseteq con(y)$).

    \textbf{Appending rules.} By rule $x_C$, the term $b:C$ is added from $b I x_C$. By induction applied to item 4, we have $con(b)\sqsubseteq C$. By rule $a_C$, $y::C$ is added from $a_C I y$ and by induction applied to item 5, we have $C\sqsubseteq con(y)$.
\end{proof}

Lemma \ref{lem:universal model properties} implies that for any consistent ABox $\mathcal{A}$, the model generated from $\mathcal{A}$ using Algorithm \ref{alg:main algo} acts as a universal model for several types of queries. We describe some such queries below.

\fakeparagraph{Relationship queries.} These queries are either Boolean queries asking if two individuals are related by  relation $I$, $R_\Box$ or $R_\Diamond$, e.g.~$q= b I y$, or queries asking for names of all individuals appearing in $\mathcal{A}$ that are related to some element by  relation $I$, $R_\Box$ or $R_\Diamond$, e.g.~$q(p)= b R_\Box p$. 

\fakeparagraph{Membership queries.} These queries are either Boolean queries asking if some object or feature belongs to a given concept, e.g.~$q= y::C$, or queries asking for names of all individuals appearing in $\mathcal{A}$ that are in the extension or intension of a concept $C$, e.g.~$q(p)=p:C$. 

\fakeparagraph{Subsumption queries.} These queries are Boolean queries asking if a concept $C_1$ is included in $C_2$, i.e.~$q=C_1 \sqsubseteq C_2$.\footnote{Note that no non-trivial subsumptions are implied by knowledge bases with acyclic TBoxes. However, we include such queries as the algorithm can be used to answer queries regarding trivial (those implied by logic) subsumption efficiently. Moreover, we believe that the algorithm  extend ideas used to answer these queries may be used in future generalizations to knowledge bases with cyclic TBoxes.}

As Algorithm \ref{alg:main algo} is polynomial-time and gives a model which is of polynomial-size in  $|\mathcal{A}|$, we have the following corollary. 

\begin{corollary}
 For any LE-$\mathcal{ALC}$  ABox $\mathcal{A}$, a query $q$ of the above forms consisting of concepts and individual names appearing in $\mathcal{A}$, can be answered in polynomial time in $|\mathcal{A}|$ using Algorithm \ref{alg:main algo}. 
\end{corollary}

\begin{remark}
    Relationship, membership, and subsumption queries can also be answered in polynomial time by converting them into a problem of consistency checking (see \cite{DLhandbook} for more details). However, it involves performing tableaux expansion for each query, while our result implies that we can answer multiple Boolean and naming queries with a single run of tableaux algorithm. 
\end{remark}

If a subsumption or membership query has concept $C$ not appearing in $\mathcal{A}$, we can answer such query by adding $C$ to $\mathcal{A}$ through creation rule i.e.~adding terms $a_C:C$, and $x_C::C$ to $\mathcal{A}$. If we have multiple queries consisting of concepts not appearing in $\mathcal{A}$, we can add all of these concepts simultaneously and answer all queries with a single run of the tableaux algorithm.  


\fakeparagraph{Disjunctive relationship and membership queries.} A disjunctive relationship (resp.~membership) query is formed by taking the 
disjunction of a finite number of relationship (resp.~membership) queries, e.g.~$q=b I y \vee b I z$, and $q(p)= p:C_1 \vee p:C_2$\footnote{The symbol $\vee$ in this paragraph refers to join in the first-order (meta) language, and not join of concepts in LE-$\mathcal{ALC}$.}. The following lemma implies that we can answer such queries in LE-$\mathcal{ALC}$ by answering each disjunct separately.

\begin{lemma}\label{lem:disjunctive queries}
Let $t_1$ and $t_2$ be any LE-$\mathcal{ALC}$ ABox terms not containing negation. Then, for any consistent LE-$\mathcal{ALC}$ ABox  $\mathcal{A}$,  $\mathcal{A} \models t_1 \vee t_2$  iff $\mathcal{A} \models t_1$ or $\mathcal{A} \models t_2 $.
\end{lemma}
\begin{proof}
    $\mathcal{A} \models t_1 \vee t_2$ iff $\mathcal{A} \cup  \{ \neg t_1,  \neg  t_2\}$ is inconsistent. By tableaux algorithm for  LE-$\mathcal{ALC}$, we have $\overline{\mathcal{A} \cup  \{ \neg t_1,  \neg  t_2\}}= \overline{\mathcal{A}} \cup  B$, {where the only terms in $B$ are $ \neg t_1$, $\neg t_2$, and the terms obtained by applying the negative rule $\neg b$ or $\neg x$ to these terms.} Therefore, as $\mathcal{A} \cup \{ \neg t_1,  \neg  t_2\}$ is inconsistent, $\overline{\mathcal{A}} \cup  B$ must contain a clash. But as $\mathcal{A}$ is consistent, $\overline{\mathcal{A}}$ does not contain any clash. Therefore, some term in $\overline{\mathcal{A}}$ clashes with $\neg t_1$ or $\neg t_2$ or the term obtained by applying negative rule $\neg b$ or $\neg x$ to these terms. This implies that $\mathcal{A} \cup\{ \neg t_1\}$ or $\mathcal{A} \cup  \{ \neg t_1\}$ must be inconsistent. Therefore, we have $\mathcal{A} \models t_1$ or $\mathcal{A} \models t_2 $. 
\end{proof}

\subsection{Negative queries}
Negative queries are obtained by applying negation to relationship, membership and subsumption queries discussed above. These queries ask if the given ABox implies that some object is not related to some feature or some object (resp.~feature) does not belong to some concept, or that one concept is not included in another concept. We start with negative relationship queries.

\begin{lemma}\label{lem:negative relation}
For any consistent LE-$\mathcal{ALC}$ ABox $\mathcal{A}$ and for any individual names $b$ and $y$,

1.~$ \mathcal{A} \models \neg (b I y)$ iff  $ \neg (b I y)\in\mathcal{A}$, 

2.~$ \mathcal{A} \models \neg (b R_\Box y)$ iff  $\neg (b R_\Box y) \in \mathcal{A}$,

3.~$\mathcal{A} \models  \neg (y R_\Diamond b)$ iff $\neg (y R_\Diamond b) \in \mathcal{A}$. 

\end{lemma}

\begin{proof}
We only prove items 1 and 2. The proof for item 3 is similar. For item 1, the right to left implication is trivial. For the left to right implication, suppose $ \mathcal{A} \models \neg (b I y)$. Then, $\mathcal{A} \cup \{b I y\}$  must be inconsistent. As $b$ and $y$ appear in $\mathcal{A}$, no tableaux expansion rule has the term $b I y$ in its premise. Therefore, the tableaux completion of $\mathcal{A} \cup \{b I y\}$ is $\overline{\mathcal{A}} \cup \{b I y\}$. As $\mathcal{A}$ is consistent, $\overline{\mathcal{A}}$ does not contain clash. Therefore, since $\overline{\mathcal{A}} \cup \{b I y\}$ must contain a clash, we have $\neg (b I y) \in \overline{\mathcal{A}}$. However, note that no tableaux expansion rule can add such term for individual names $b$, $y$ appearing in the original ABox $\mathcal{A}$. Therefore, $\neg (b I y) \in {\mathcal{A}}$.

For item 2, the right to left implication is also trivial. For the left to right implication, suppose $\mathcal{A} \models \neg (b R_\Box y)$. Then, $\mathcal{A} \cup \{b R_\Box y\}$  must be inconsistent. As $b$ and $y$ appear in $\mathcal{A}$, the only tableaux expansion rule having term $b R_\Box y$ in its premise is adjunction rule $R_\Box$ which adds terms $\Diamondblack b I y$ and $b I \Box y$. Again, the only rules that have any of these terms in premise are $I$-compatibility rules that add $b R_\Box y$ to the tableaux expansion. Therefore, the tableaux completion of $\mathcal{A} \cup \{b R_\Box y\}$ is $\overline{\mathcal{A}} \cup \{b R_\Box y, \Diamondblack  b I y, b I \Box y\}$. As $\mathcal{A}$ is consistent, $\overline{\mathcal{A}}$ does not contain a clash. Therefore, as $\overline{\mathcal{A}}\cup\{b R_\Box y, \Diamondblack b I y, b I \Box y\}$ must contain a clash, one of the terms $\neg (b R_\Box y)$, $\neg (\Diamondblack  b I y) $ or $\neg ( b I \Box y)$ must be in $\overline{\mathcal{A}}$. However, no expansion rule can add the terms of any of these forms. Furthermore, the terms of the form $\neg (\Diamondblack b I y) $ or $\neg (b I \Box y)$ cannot appear in the original ABox $\mathcal{A}$. Therefore, $\neg (b R_\Box y)$ must be in $\mathcal{A}$. 
\end{proof}

As a result, we can answer negative relationship queries over a consistent  LE-$\mathcal{ALC}$ ABox $\mathcal{A}$ in linear time by searching through $\mathcal{A}$. 

We cannot apply a similar strategy to membership queries of the form $\neg (b:C)$ or $\neg (y::C)$, as such terms can be implied by $\mathcal{A}$ without being present in $\overline{\mathcal{A}}$. For example, consider ABox $\mathcal{A} = \{b:C_1, \neg (b:C_1 \wedge C_2)\}$ which implies $\neg (b:C_2)$, but this term does not appear in 
$\overline{\mathcal{A}}$. This means that the model obtained by tableaux algorithm is not a universal model for these types of queries. Hence, to answer queries of the form $\neg (b:C)$ or $\neg (y::C)$, we must proceed by the usual route of adding the terms $b:C$ or $y::C$ to $\mathcal{A}$ and checking the consistency of the resulting ABox.
We can also consider negative subsumption queries, i.e.~queries asking whether the given ABox $\mathcal{A}$ implies one concept $C_1$ is not included in another concept $C_2$, denoted as $\neg (C_1\sqsubseteq C_2)$.
Answering this query is the same as answering if the knowledge base obtained by adding the TBox axiom $C_1\sqsubseteq C_2$ to ABox $\mathcal{A}$  is consistent. We can answer these queries for any TBox term $C_1\sqsubseteq C_2$, such that no sub-formula of $C_1$ appears in $C_2$, by using unraveling on $C_1\sqsubseteq C_2$, and then applying Algorithm \ref{alg:main algo}. 

In this and previous sections, we have discussed answering Boolean queries of all the forms which an LE-$\mathcal{ALC}$ ABox term can take. Hence, we can combine these methodologies to answer ontology {\em equivalence queries} asking if two ABoxes are equivalent, i.e. $\mathcal{A}_1 \equiv \mathcal{A}_2$, by checking if every term in $\mathcal{A}_2$ is implied by $\mathcal{A}_1 $ and vice versa.

\subsection{Separation and differentiation queries}\label{ssuc:separation queries}
An important set of queries is queries  asking if the given knowledge base implies (ensures) that two individuals can be differentiated from each other by a certain property. In this section, we consider some queries of this type in LE-$\mathcal{ALC}$. 

{\em Separation queries} are queries of the form $S(b,d)=\exists p (b I p \wedge \neg (d I p))$ or $S(y,z)= \exists p (p I y \wedge \neg (p I z)) $  for two object (resp.~feature) names $b$, $d$ (resp.~$y$, $z$) appearing in a given ABox. These queries can be understood as asking whether two given objects or features can be separated for sure using relation $I$ based on the given knowledge base. Note that for any 
LE-$\mathcal{ALC}$ ABox $\mathcal{A}$,

1.~$\mathcal{A} \not \models \exists p(b I p \wedge \neg (d I p))$  iff $\mathcal{A}  \cup \{\forall p(b I p \Rightarrow d I p)\}$ is consistent, and 

2.~$\mathcal{A} \not \models \exists p (p I y \wedge \neg (p I z))$  iff $\mathcal{A}  \cup \{\forall p (p I y \Rightarrow p I z)\}$ is consistent. 

Therefore, a separation query $S(b,d)$ (resp.~$S(y,z)$) can be answered by checking if $\mathcal{A}$ is consistent in the extension of LE-$\mathcal{ALC}$ with the axiom $\forall p (b I p \Rightarrow d I p)$ (resp.~$\forall p (p I y \Rightarrow p I z)$). To this end, we consider the expansion of the LE-$\mathcal{ALC}$ tableaux algorithm with the rules 

{{\centering
\begin{tabular}{cc}
\AXC{$b I x$}
\LL{\fns $SA(b,d)$}
\UIC{$d I x$}
\DP 
\ & \
\AXC{$a I y$}
\RL{\fns $SX(y,z)$}
\UIC{$a I z$}
\DP.
\end{tabular}
\par
}}

\vspace{-0.1cm}
\begin{theorem}\label{thm:seperation}
The tableaux algorithm obtained by adding the rule $SA(b,d)$ (resp.~$SX(y,z)$) to the LE-$\mathcal{ALC}$ tableaux expansion rules provides a polynomial-time sound and complete decision procedure for checking the consistency of $\mathcal{A}  \cup \forall p (b I p \Rightarrow d I p) $(resp.~$\mathcal{A}  \cup \forall p(p I y \Rightarrow p I z)$). 
\end{theorem}

Now, we will prove the termination, soundness and completeness of tableaux algorithms for checking consistency of $\mathcal{A}  \cup \{\forall p (b I p \Rightarrow d I p)$\}  and  $\mathcal{A}  \cup \{\forall p (p I y \Rightarrow p I z)\}$ defined in Section \ref{ssuc:separation queries}. We only give the proof for $\mathcal{A} \cup \forall p (b I p \Rightarrow d I p)$, the proof for $\mathcal{A}  \cup \forall p (p I y \Rightarrow p I z)$ would be similar.

In this paper, we will only explain the changes that must be made to the termination, soundness, and completeness proofs of the LE-$\mathcal{ALC}$ tableaux algorithm provided in \cite{van2023non}. We refer to \cite{van2023non} for details of these proofs.

\fakeparagraph{Termination.} To prove termination, we prove that the following lemma proved for LE-$\mathcal{ALC}$ tableaux algorithm in \cite{van2023non} also holds for its extension with rule $SA(b,d)$. 

\begin{lemma} \cite[Lemma 1]{van2023non}\label{lem:depth}
For any individual names $b$, and $y$, and concept $C$  added during tableau expansion of ${\mathcal{A}}$, 
\begin{equation}\label{eq:IH 1}
 \Box_{\mathcal{D}}(C) \leq  \Box_{\mathcal{D}}(\mathcal{A})+1 \, 
 \mbox{ and } \, \Diamond_{\mathcal{D}}(C) \leq  \Diamond_{\mathcal{D}}(\mathcal{A})+1,  
\end{equation}

\begin{equation}\label{eq:IH 2}
-\Diamond_{\mathcal{D}}(\mathcal{A} )-1 \leq \Diamond_{\mathcal{D}} (b) \, \mbox{ and } \,   \Box_{\mathcal{D}}(b) \leq  \Box_{\mathcal{D}}(\mathcal{A})+1, 
\end{equation}

\begin{equation}\label{eq:IH 3}
-\Box_{\mathcal{D}}(\mathcal{A})-1 \leq \Box_{\mathcal{D}}(y)\, \mbox{ and } \, \Diamond_{\mathcal{D}}(y) \leq  \Diamond_{\mathcal{D}}(\mathcal{A})+1
\end{equation}
\end{lemma}
\begin{proof}
 The proof proceeds by showing that the following stronger claim holds. For any tableaux expansion  $ \overline{\mathcal{A}}$, obtained from $\mathcal{A}$ after any finite number of expansion steps:

1.~For any term $b I y \in \overline{\mathcal{A}}$, $\Box_{\mathcal{D}} (b) -  \Box_{\mathcal{D}} (y) \leq \Box_{\mathcal{D}}(\mathcal{A})+1$, and $ \Diamond_{\mathcal{D}} (y) -  \Diamond_{\mathcal{D}} (b) \leq \Diamond_{\mathcal{D}}(\mathcal{A})+1$.

2.~For any term $b R_\Box y \in \overline{\mathcal{A}}$, $\Box_{\mathcal{D}} (b) +1 -  \Box_{\mathcal{D}} (y) \leq \Box_{\mathcal{D}}(\mathcal{A})+1$, and $ \Diamond_{\mathcal{D}} (y) -  \Diamond_{\mathcal{D}} (b) \leq \Diamond_{\mathcal{D}}(\mathcal{A})+1$.

3.~For any term $y R_\Diamond b \in \overline{\mathcal{A}}$, $\Box_{\mathcal{D}} (b) -  \Box_{\mathcal{D}} (y) \leq \Box_{\mathcal{D}}(\mathcal{A})+1$, and $ \Diamond_{\mathcal{D}} (y) +1  -  \Diamond_{\mathcal{D}} (b) \leq \Diamond_{\mathcal{D}}(\mathcal{A})+1$.

4.~For any term $b:C \in \overline{\mathcal{A}}$, $ \Box_{\mathcal{D}} (b)  +  \Box_{\mathcal{D}} (C) \leq \Box_{\mathcal{D}}(\mathcal{A})+1$, and $ -\Diamond_{\mathcal{D}} (b) -  \Diamond_{\mathcal{D}} (C) \leq 0$.

5.~For any term $y::C \in \overline{\mathcal{A}}$, $ -\Box_{\mathcal{D}} (y)  - \Box_{\mathcal{D}} (C) \leq 0$, and $ \Diamond_{\mathcal{D}} (y) +  \Diamond_{\mathcal{D}} (C) \leq \Diamond_{\mathcal{D}}(\mathcal{A})+1$.

The proof proceeds by induction of number of rules applied. The proofs for initial case (i.e.~for all the terms in original ABox $\mathcal{A}$) and all the LE-$\mathcal{ALC}$  tableaux expansion rule are provided in \cite[Lemma 1]{van2023non}. Therefore, to complete the proof we need to show that the rule  $SA(b,d)$ also preserves these properties. Suppose a term $d I y$ is added from a term $b I y$ using rule $SA(b,d)$. In this case, by induction, we have $\Box_{\mathcal{D}} (b) -  \Box_{\mathcal{D}} (y) \leq \Box_{\mathcal{D}}(\mathcal{A})+1$, and $ \Diamond_{\mathcal{D}} (y) -  \Diamond_{\mathcal{D}} (b) \leq \Diamond_{\mathcal{D}}(\mathcal{A})+1$. As $b$ and $d$ are object names appearing in $\mathcal{A}$, we have $\Box_{\mathcal{D}} (b) = \Box_{\mathcal{D}} (d)=\Diamond_{\mathcal{D}} (b) = \Diamond_{\mathcal{D}} (d)=0$. Therefore, we have $\Box_{\mathcal{D}} (d) -  \Box_{\mathcal{D}} (y) \leq \Box_{\mathcal{D}}(\mathcal{A})+1$, and $ \Diamond_{\mathcal{D}} (y) -  \Diamond_{\mathcal{D}} (d) \leq \Diamond_{\mathcal{D}}(\mathcal{A})+1$. Hence proved.
\end{proof}

This lemma bounds the number of new contestant and concept names that can appear in tableaux expansion. Therefore, it implies that the number of  terms that can appear in tableaux expansion are bounded by $poly(size(\mathcal{A}))$. As tableaux has no branching rules, this implies termination. (See \cite{van2023non} for more details).

\paragraph{Soundness.} The soundness follows immediately from the soundness for LE-$\mathcal{ALC}$ tableaux algorithm \cite[Section 4.2]{van2023non}, and the fact that  the rule $SA(b,d)$ ensures that the model obtained from completion satisfies the axiom $\forall p (b I p \Rightarrow d I p)$.

\paragraph{Completeness.} To prove completeness, we show that the following lemma proved for LE-$\mathcal{ALC}$ \cite[Lemma 5]{van2023non} also holds for its extension with the axiom $\forall p (b I p \Rightarrow d I p)$. 

\begin{lemma}\label{lem:characteristic consistency}
For any ABox $\mathcal{A}$,  any model $M=(\mathbb{F}, \cdot^{\mathcal{M}})$ of $\mathcal{A}   $ can be extended to a model $M'=(\mathbb{F}', \cdot^{\mathcal{M}'})$ such that $\mathbb{F}'=(A',X',I',\{R_\Box'\}_{\Box\in\mathcal{G}}, \{R_\Diamond'\}_{\Diamond\in\mathcal{F}})$, $A \subseteq A'$ and $X \subseteq  X'$, and moreover for every $\Box \in \mathcal{G}$ and $\Diamond \in \mathcal{F}$:

1. There exists $a_C \in A'$ and $x_C \in X'$  such that:
\begin{equation}\label{eq:completness 1}
   C^{\mathcal{M}'} =(I'^{(0)}[x_C^{\mathcal{M}'}], I'^{(1)}[a_C^{\mathcal{M}'}]), \quad a_C^{\mathcal{M}'} \in \val{C^{\mathcal{M}'}}, \quad  x_C^{\mathcal{M}'} \in \descr{C^{\mathcal{M}'}},  
\end{equation}

 2. For every individual $b$ in $A$ there exist $\Diamond b$ and $\Diamondblack b$ in $A'$ such that:
\begin{equation}\label{eq:completness 2}
I'^{(1)}[\Diamondblack b] = R_\Box'^{(1)}[b^{\mathcal{M}'}] \quad \mbox{and} \quad I'^{(1)}[\Diamond b] = R_\Diamond'^{(0)}[b^{\mathcal{M}'}],
\end{equation}

3. For every individual $y$ in $X$ there exist $\Box y$ and $\blacksquare y$ in $X'$ such that:
\begin{equation}\label{eq:completness 3}
I'^{(0)}[\blacksquare y] = R_\Diamond'^{(1)}[y^{\mathcal{M}'}] \quad \mbox{and} \quad I'^{(0)}[\Box y] = R_\Box'^{(0)}[y^{\mathcal{M}'}]. 
\end{equation}

4. For any $C$, $\val{C^\mathcal{M}}= \val{C^{\mathcal{M}'}} \cap A$ and $\descr{C^\mathcal{M}}= \descr{C^{\mathcal{M}'}} \cap X$. 
\end{lemma}
In \cite[Section 4.3]{van2023non}, this lemma was proved by constructing such model $M'$.  Here, we show that if the model $M$ additionally satisfies axiom $\forall p (b I p \Rightarrow d I p)$, then so does the model $M'$. This follows from the fact that the model $M'$ is constructed in such a way that (see \cite[Section 4.3]{van2023non} for more details) for any $b,d \in A$, a term $b I^{\mathcal{M}'} x_C$ is added for some newly added element $x_C\in X'\setminus X$ iff we have $b I^{\mathcal{M}} y$ for all $y \in C^{\mathcal{M}}$. Then, as $M \models \forall p (b I p \Rightarrow d I p)$, we also get $d I y$ for all $y \in C^{\mathcal{M}}$  which implies $d I^{\mathcal{M}'} x_C$. 

The above lemma ensures that if $\mathcal{A}  \cup \{\forall p (b I p \Rightarrow d I p)$\}  has a model, then it has a model with classification of objects and features. The completeness proof then proceeds by showing that if $\mathcal{A} $ is consistent, then the model for $\mathcal{A} $ with the above properties satisfies all the terms in the completion of $\mathcal{A} $. We refer to \cite[Section 4.3]{van2023non} for details.


Hence, we can use this expanded tableaux algorithm to answer separation queries over a given ABox $\mathcal{A}$ in polynomial time. This also allows us to answer {\em differentiation queries} $Dif(b,d) = S(b,d) \wedge S(d,b)$ (resp.~$Dif(y,z) = S(y,z) \wedge S(z,y)$) which ask if $\mathcal{A}$ implies that $b$ and $d$ (resp.~$y$ and $z$) can be differentiated from each other by the relation $I$ in polynomial time. We can similarly define and answer the separation and differentiation queries for the relations $R_\Box$ and $R_\Diamond$. Furthermore, we can answer {\em identity queries} asking if $\mathcal{A}$ implies that two individuals are not identical by checking if they can be differentiated by some relation.

The strategy used to answer separation queries can also be used to answer many other interesting types of queries. For example, we can consider separation queries which ask about separation between different relations. Consider the queries $SA(R_\Box, R_\Diamond, b) = \exists p (b R_\Box  p \wedge \neg (p R_\Diamond b))$, and $SX(R_\Diamond, R_\Box, y) = \exists p (y R_\Diamond  p \wedge \neg (p R_\Box y))$. These queries ask if the given ABox implies that the relations $R_\Box$ and $R_\Diamond$ are the local inverses of each other in some object $b$ or feature $y$ appearing in the given ABox. 

Note that for any LE-$\mathcal{ALC}$ ABox $\mathcal{A}$,

1.~$\mathcal{A} \not \models  \exists p (b R_\Box  p \wedge \neg (p R_\Diamond b))$  iff $\mathcal{A}  \cup \{\forall p(b R_\Box  p  \Rightarrow  p R_\Diamond b)\}$ is consistent, and 

2.~$\mathcal{A} \not \models \exists p (y R_\Diamond  p \wedge \neg (p R_\Box y))$  iff $\mathcal{A}  \cup \{\forall p (y R_\Diamond  p \Rightarrow  p R_\Box y)\}$ is consistent. 

Therefore, a separation query $SA(R_\Box, R_\Diamond, b) $ (resp.~$SX(R_\Diamond, R_\Box, y)$) can be answered by checking if $\mathcal{A}$ is consistent in the extension of LE-$\mathcal{ALC}$  with the axiom $\forall p(b R_\Box  p  \Rightarrow  p R_\Diamond b)$ (resp.~$\forall p (y R_\Diamond  p \Rightarrow  p R_\Box y)$). To this end, we consider the expansion of the LE-$\mathcal{ALC}$ tableaux algorithm with the following rules 

\smallskip
{{\centering
  \begin{tabular}{cccc}
\AXC{$ b R_\Box y$}
\LL{\fns $SA(R_\Box, R_\Diamond, b) $}
\UIC{$y R_\Diamond b$}
\DP 
\ & \ 
\AXC{$y R_\Diamond b$}
\RL{\fns $SX(R_\Diamond, R_\Box, y)$}
\UIC{$b R_\Box y$}
\DP. 
\end{tabular}
\par
}}

\begin{theorem}\label{thm:seperation 2}
The tableaux algorithm obtained by adding rule $SA(R_\Box, R_\Diamond, b)$ (resp.~$SX(R_\Box, R_\Diamond, y)$) to LE-$\mathcal{ALC}$ tableaux expansion rules provides a polynomial-time  sound and complete decision procedure for checking consistency of $\mathcal{A}  \cup  \{\forall p(b R_\Box  p  \Rightarrow  p R_\Diamond b)\} $(resp.~$\mathcal{A}  \cup \{\forall p (y R_\Diamond  p \Rightarrow  p R_\Box y)\}$). 
\end{theorem}
\begin{proof}
See Appendix \ref{app:B}. 
\end{proof}

We can similarly answer the queries of the forms  $SA(R_\Box, I, b) = \exists p (b R_\Box  p \wedge \neg (b I p))$, and $SX( R_\Diamond, I, y) = \exists p (y R_\Diamond  p \wedge \neg (p I y))$ using tableaux algorithm expanded with the rules

{{\centering
\begin{tabular}{cccc}
\AXC{$ b R_\Box y$}
\LL{\fns }
\UIC{$b I y$}
\DP 
\ &and& \ 
\AXC{$y R_\Diamond b$}
\RL{\fns }
\UIC{$b I y$}
\DP.
\end{tabular}
\par}}

However, this does not apply to all separation queries on relations. For example, consider queries of the form $SA(I, R_\Box, b) = \exists p (b I  p \wedge \neg (b R_\Box p))$, and $SX(I, R_\Diamond, y) = \exists p (p I y \wedge \neg (y R_\Diamond  p))$. This is because the expansion of LE-$\mathcal{ALC}$ tableaux algorithm with the rules 

{{\centering
\begin{tabular}{cccc}
\AXC{$b I y$}
\LL{\fns }
\UIC{$ b R_\Box y$}
\DP 
\ &and& \
\AXC{$b I y$}
\RL{\fns }
\UIC{$y R_\Diamond b$}
\DP 
\end{tabular}
\par}}

\noindent may not be terminating. We leave this as part of future work.

\section{Examples}\label{sec:examples}
In this section, we give a toy example of an LE-$\mathcal{ALC}$ knowledge base and some  queries of different types to demonstrate  working of various algorithms for  answering these queries discussed in the paper.

Suppose, we want to create a knowledge base to represent categorization of some movies on a streaming website  which can be used to answer some queries based on them.  We list the concept names and individual names for objects and features appearing in the knowledge base in the following tables:

{{
\centering\footnotesize
 \begin{tabular}{|c|c|c|c|c|c|}
    \hline
        \textbf{concept name} & \textbf{symbol} & 
        \textbf{concept name} & \textbf{symbol} & 
        \textbf{concept name} & \textbf{symbol} \\
        \hline
          Italian movies & $IM$ & German movies & $GM$ & French movies & $FM$\\  European movies & $EUM$ & Recent movies & $RM$  & Recent drama movies & $RDM$ \\  Drama movies & $DM$  & Famous drama movies & $FDM$ &&\\
        \hline
    \end{tabular}
    \par 
}}
\smallskip

{{
\centering\footnotesize
 \begin{tabular}{|c|c|c|c|}
     \hline
         \textbf{object} & \textbf{symbol} & 
         \textbf{object} & \textbf{symbol}\\ 
         \hline
    All the President's Men & $m_1$ & Spirited Away & $m_2$ \\ Oppenheimer & $m_3$
    & Cinema Paradiso  & $m_4$ \\
    \hline
 \end{tabular}

 \smallskip
 
 \begin{tabular}{|c|c|c|c|c|c|}
    \hline
        \textbf{feature} & \textbf{symbol} & \textbf{feature} & \textbf{symbol}  & \textbf{feature} & \textbf{symbol}  \\
        \hline
    German language  & $f_1$ &  French language  & $f_2$ & Based on real story & $f_3$  \\  Serious plot & $f_4$ &  Released after 2015 & $f_5$ & Released after 2020 & $f_6$ \\
          \hline 
           \end{tabular}
           \par
}}
\smallskip

Suppose  the following  knowledge base $\mathcal{K}_1$ presents the information obtained by  the website  regarding the movies, their features and their categorization into above categories from some initial source which possibly has incomplete information. 

\begin{multicols}{2}
\noindent
\begin{align*}
    \mathcal{A}_1=\{& m_4:IM, \neg(m_4 I x_\emptyset), x_{\emptyset} :: FM\wedge IM,\\ & x_{\emptyset}::GM\wedge IM, f_1::GM, f_2::FM,  \\ & f_4::DM, m_3:RDM, m_3 I f_6, m_1 I f_3, \\ & \neg(m_1 I f_2), \neg(m_2:EUM)\},
\end{align*}

\columnbreak
\noindent
\begin{align*}
    \mathcal{T}_1=\{EUM&\equiv GM\vee FM,\\
    RDM&\equiv RM\wedge DM, \\ IM&\sqsubseteq EUM\}.
\end{align*}
\end{multicols}

Since TBox $\mathcal{T}_1$ is acyclic, we can convert the knowledge base $\mathcal{K}_1$ into an equivalent knowledge base with only ABox using unraveling, which would be suitable for our algorithms.

For any movie $m$ and feature $y$, we have $m I y$ (resp.~$\neg(m I y)$) iff according to the initial source database, movie $m$ has (resp.~does not have) feature $y$. The feature $x_{\emptyset}$ intuitively represents a contradiction. The terms $x_{\emptyset}:: FM\wedge IM$ and $x_{\emptyset}:: GM\wedge IM$  states that there is no movie that is both a French movie and Italian movie or both a German and Italian movie. The term $m_4:IM$ specifies that Cinema Paradiso is an Italian movie. The term $f_3::AM$ states that Action movies have action sequences. Other terms in $\mathcal{A}_1$ can be explained similarly. The term $EUM\equiv GM\vee FM$ states that the category of European movies is the smallest category on the website which contains both German movies and French movies. The term $IM\sqsubseteq EUM$ can be equivalently written as $IM\equiv EUM\wedge C$ for some new category $C$, meaning  all Italian movies are European movies. Other terms in $\mathcal{T}_1$ can be explained similarly. Note that the terms $\neg(m_4 I x_{\emptyset})$, $x_{\emptyset} :: FM\wedge IM$, and  $x_{\emptyset}::GM\wedge IM$ together imply that $m_4$ is not in $(FM\wedge IM)\vee(GM\wedge IM)$. However, $m_4$ is in $IM = EUM\wedge IM = (GM\vee FM) \wedge IM$. Therefore, this knowledge base is inconsistent in distributive logic but it is consistent in our setting of LE-$\mathcal{ALC}$.

 Additionally, the website also tries to get an understanding of  subjective (epistemic) view  of different user groups  on the website regarding movies, their features, and categorization.  To this end, website asks some  users from different groups the following two questions:
 
 \noindent (a) Given a list of movies:
 
 (a1) Please choose movies which have feature $y$ from the list; 
 
 (a2) please choose movies which do not have feature $y$ from the list.
 
 \noindent (b) Given a list of features:
 
  (b1) please choose features that describe movie $m$ from the list;
 
 (b2) please choose features which do not describe movie $m$ from the list. 
 
 Note that there can be movies (resp.~features) in the list of options which are not chosen as answer to either (a1) or (a2) (resp.~(b1) or (b2)). 
 
 We model information  obtained from above questions as follows: If some user from group $i$ chooses movie $m$ (resp.~movie $m$, resp.~feature $y$, resp.~feature $y$) as an answer  to question (a1) (resp.~(a2), resp.~(b1), resp.~(b2)), then we add  $m {R_\Box}_i y$ (resp.~$\neg(m {R_\Box}_i y)$, resp.~$y {R_\Diamond}_i m$, resp.~$\neg(y {R_\Diamond}_i m)$) to the knowledge base. Note that, in general,  none of the terms $y {R_\Diamond}_i m$, and  $m {R_\Box}_i y$ implies  other. This is because question (a) and question (b) may be asked to different users from group $i$.
 
 Then, for any category $C$, $[{R_\Box}_i] C$ denotes the category defined by objects which are reported to have  all the features in $\descr{C}$ (description of $C$) by some user in group $i$. Thus, $[{R_\Box}_i] C$ can be seen as the category of movies which are considered to be in $C$ according to the user group $i$. This means that for any movie $m$ in $[{R_\Box}_i] C$ and any feature $y$ of $C$, some user in the group $i$ will name $m$ as a movie having feature $y$ as an answer to (a1).
 
 Similarly, $\langle {R_\Diamond}_i\rangle C$ denotes the  category defined by  features  which all objects in $\val{C}$ (objects in $C$) are reported to have by some user in group $i$. Thus, $\langle {R_\Diamond}_i\rangle C$ can be seen as the category of movies defined by features which are considered to be in the description of $C$ according to the user group $i$. This means that  for any feature $y$ of $\langle {R_\Diamond}_i\rangle C$ and any movie $m$ in $C$, some user in the group $i$ will name $y$ as a feature of the movie $m$ as an answer to (b1).

 Moreover, the website can also  ask users the following questions:
 
 (c) Please choose movies which belong to the category $C$ from a given list of movies. 
 
 (d) Please choose features that describe the category $C$ from a given list of features.
 
 If $m$ (resp.~$y$) is chosen as an answer of (c) (resp.~(d)) by some user in group $i$, then we add term $m:[{R_\Box}_i] C$ (resp.~$y::\langle {R_\Diamond}_i\rangle C$) to the knowledge base. Here, we assume that for any feature $y$ in description of $C$, if some user chooses $m$ as a movie in category $C$, then there is some user from the same group who will also choose $m$ as a movie with feature $y$ in answer to (a1). This assumption ensures that the $[{R_\Box}_i] C$ is interpreted  in  accordance with LE-$\mathcal{ALC}$ semantics from relation ${R_\Box}_i$.
 
 Similarly,  for any movie $m$ in $C$,  we assume that if some user chooses $y$ as a feature in category $C$, then there is some user from the same group will also choose $y$ as a feature with movie $m$ in answer to (b1). 
  This assumption ensures that the  $\langle {R_\Diamond}_i\rangle C$  is interpreted  in  accordance with LE-$\mathcal{ALC}$ semantics from relation ${R_\Diamond}_i$ \footnote{These assumptions can be justified if we assume we have a large number of users in each group so that at least some users in the group will have the information regarding all the movies and their features under consideration.}.

The following table presents knowledge  base $\mathcal{K}_2$  representing different user groups' views regarding the movies obtained from the answers to the above questions. For simplictiy, we assume that we have only two different user groups. From here on, we will use $\Box_i$ (resp.~$\Diamond_i$) to denote $[{R_\Box}_i]$ (resp.~$\langle {R_\Diamond}_i \rangle$) for $i=1,2$.
\begin{multicols}{2}
\noindent
\begin{align*}
    \mathcal{A}_2=\{& m_3 {R_\Box}_1 f_3, m_3 {R_\Box}_2 f_3,\neg(m_1 {R_\Box}_1 f_6),\\ & f_3 {R_\Diamond}_1 m_3, f_3 {R_\Diamond}_2 m_3, m_3: \Box_2 RDM, \\ & f_5::\Diamond_1 RM, \neg(m_1 {R_\Box}_2 f_5)\},
\end{align*}
\columnbreak
\noindent
\begin{align*}
    \mathcal{T}_2=\{FDM &\equiv \Box_1 DM \wedge \Box_2 DM\}.
\end{align*}
\end{multicols}
Let $\mathcal{K} =\mathcal{K}_1 \cup \mathcal{K}_2 $
be the knowledge base obtained by combining knowledge from the source database and from users.
Given the knowledge base $\mathcal{K}$, we can answer the following queries.

\fakeparagraph{Positive queries.}
By Lemma \ref{lem:universal model properties}, these queries can be answered using the universal model constructed using the Tableaux Algorithm~\ref{alg:main algo} from the ABox obtained by unraveling $\mathcal{K}$. We depict this model in Appendix \ref{App:Models for example knowledge base}.  

(1) $q(p)=m_3 I p$ asking to name all the features implied by $\mathcal{K}$ that the movie Oppenheimer has. Using the universal model, we can give the answer $q(a)=\{f_4,f_6\}$.

(2) $q=m_4:FDM$ asking if $\mathcal{K}$ implies that Cinema Paradiso is a Famous drama movie. We can give  answer `No' since in the universal model, $I(m_4, x_{FDM})=0$. 

(3) $q=\Box_2 RDM \sqsubseteq  \Box_2 DM$ asking if $\mathcal{K}$ implies that all the movies considered to be recent drama movies by users in group 2 are also considered to be drama movies by them. We can give  answer `Yes' since in the universal model, $a_{\Box_2 RDM} I x_{\Box_2 DM}$.

(4) $q=m_3:\Box_2\Diamond_1 RM$ asking if $\mathcal{K}$ implies whether for any feature which is considered to be in the description of Recent movies according to some user in group 1, there is some user in group 2 who considers (reports) Oppenheimer to have this feature. We can give answer `No' since in the universal model of the knowledge base  $ \mathcal{K}'= \mathcal{K} \cup \{a_{\Box_2\Diamond_1 RM}: \Box_2\Diamond_1 RM, x_{\Box_2\Diamond_1 RM}:: \Box_2\Diamond_1 RM$, $I(m_3, x_{\Box_2\Diamond_1 RM})=0$.  In Appendix \ref{App:Models for example knowledge base}, we provide the universal model for $\mathcal{K}$  which is obtained from  complete tableaux expansion  of $\mathcal{K}$. It is easy to check from the shape of the tableaux expansion rules  that no tableaux expansion rule can add term $m_3 I x_{\Box_2\Diamond_1 RM}$ during the expansion of knowledge base $ \mathcal{K}'$. Hence, 
$I(m_3, x_{\Box_2\Diamond_1 RM})=0$ holds in the universal model of $\mathcal{K}'$. 

\fakeparagraph{Negative queries.}
(1) $q=\neg(m_1:\Box_2\Diamond_1 RM)$ asking if $\mathcal{K}$ implies that the movie All the President's Men is not a movie in category $\Box_2\Diamond_1 RM$ (interpretation of this category is mentioned in previous example). We can give the answer `Yes' since, if we add the term $m_1:\Box_2\Diamond_1 RM$ to $\mathcal{K}$, this term along with the term $f_5::\Diamond_1 RM$ appearing in $\mathcal{K}$, we would get $m_1 {R_\Box}_2 f_5$ by rule $\Box$. Thus, the resulting knowledge base is not consistent, which means that $\mathcal{K}$ implies $\neg(m_1:\Box_2\Diamond_1 RM)$. 

(2) $q=\neg(m_3 {R_\Box}_1 f_4)$ asking if $\mathcal{K}$ implies that some user in group 1 considers Oppenheimer to be a movie which does not have a serious plot. We can give the answer `No' since the ABox obtained by  unraveling $\mathcal{K}$ does not contain the term $\neg(m_3 {R_\Box}_1 f_4)$ and by Lemma~\ref{lem:negative relation} it is not implied by $\mathcal{K}$. 

\fakeparagraph{Separation queries.}
(1) $q= Dif(m_2, m_4)$ asking if $\mathcal{K}$ implies that there is a feature that one of the movies Spirited Away and Cinema Paradiso has but the other does not. We can give the answer `Yes'. If we add the rules $SA(m_2,m_4)$ and $SA(m_4,m_2)$ to the LE-$\mathcal{ALC}$ tableaux expansion rules and run the resulting tableaux algorithm on ABox obtained by unraveling $\mathcal{K}$, we will get the clash as showed below. 

\smallskip
{{
\centering
\begin{tabular}{ccc}
\hline
   \textbf{rules}  & \textbf{premises}  & \textbf{added terms } \\
   \hline
    create & & $x_{GM\vee FM}:: GM\vee FM$\\
    $\wedge_A$ & $m_4:(GM\vee FM)\wedge C$ & $m_4:GM\vee FM$, $m_4:C$ \\
    $I$ & $m_4:GM\vee FM$,  $x_{GM\vee FM}:: GM\vee FM$ &  $m_4 I x_{GM\vee FM}$\\
    $SA(m_4,m_2)$  & $m_4 I x_{GM\vee FM}$ & $m_2 I x_{GM\vee FM}$ \\
    $\neg x$ & $\neg(m_2:GM\vee FM)$ & $\neg(m_2 I x_{GM\vee FM})$ \\
    \hline
\end{tabular}
\par
}}
\smallskip

(2) $q=SA({R_\Box}_1, I, m_4)$ asking if $\mathcal{K}$ implies that there is a feature that some user in group 1 considers Cinema Paradiso has but according to the initial source database it does not. We can give the answer `No', since if we add the rule $SA({R_\Box}_1, I, m_4)$ to the LE-$\mathcal{ALC}$ tableaux expansion rules and run the resulting tableaux algorithm on ABox obtained by unraveling $\mathcal{K}$, we will get no clash, i.e.~$\mathcal{K}$ is consistent in the extension of LE-$\mathcal{ALC}$ with the axioms $\forall y (m_4 {R_\Box}_1 y \Rightarrow m_4 I y)$.

\section{Conclusion and future work}
\label{Sec: conclusions}
In this paper, we have shown that the tableaux algorithm for LE-$\mathcal{ALC}$, or its extension with appropriate rules, can be used to answer several types of queries over LE-$\mathcal{ALC}$ ABoxes in polynomial time. Additionally, can generalize these algorithms to exponential time algorithms for  LE-$\mathcal{ALC}$ knowledge bases with acyclic TBoxes by unraveling. 

\fakeparagraph{Dealing with cyclic TBoxes and  RBox axioms.} 
In this paper, we introduced a tableaux algorithm only for knowledge bases with acyclic TBoxes. In the future, we intend to generalize the algorithm to deal with cyclic TBoxes as well. Another interesting avenue of research is to develop tableaux and query answering algorithms for extensions of LE-$\mathcal{ALC}$ with RBox axioms. RBox axioms are used in description logics to describe the relationship between different relations in  knowledge bases and  the properties of these relations such as reflexivity, symmetry, and transitivity. It would be interesting to see if it is possible to obtain necessary and/or sufficient conditions on the shape of RBox axioms for which a tableaux algorithm can be obtained. This has an interesting relationship with the problem in LE-logic of providing computationally efficient proof systems for various extensions of LE-logic in a modular manner \cite{greco2016unified,ICLArough}. 

\fakeparagraph{Universal models for other types of queries}
In this work, we showed that the model constructed from tableaux Algorithm \ref{alg:main algo} acts as universal model for several types of positive queries. In the future, it would be interesting to study if we can develop tableaux algorithms in such way that the resulting models can act as universal models for negative queries and other types of queries. This would allow the algorithm to answer multiple such queries efficiently. 

\fakeparagraph{Answering more types of queries}
In Section \ref{ssuc:separation queries}, we mentioned that  certain separation queries cannot be answered by our method due to potential non-termination of tableaux arising from the naive extension LE-$\mathcal{ALC}$ tableaux expansion rules corresponding to these queries. However, it may be possible to achieve termination in some of these case by incorporating appropriate loop check conditions into these expansion rules. In the future, we intend to study such extensions. 

\fakeparagraph{Generalizing to more expressive description logics.} The LE-$\mathcal{ALC}$ is the non-distributive counterpart of $\mathcal{ALC}$. A natural direction for further research is to explore the non-distributive counterparts of extensions of $\mathcal{ALC}$ such as $\mathcal{ALCI}$ and $\mathcal{ALCIN}$ and fuzzy generalizations of such description logics. This would allow us to express more constructions like concepts generated by an object or a feature, which can not be expressed in LE-$\mathcal{ALC}$. This would provide us language to answer many more types of interesting queries regarding enriched formal contexts. 

\bibliographystyle{eptcs}
\bibliography{generic}

\appendix


\section{Proof of Theorem \ref{thm:seperation 2}}\label{app:B}
We only give the proof for $SA(R_\Box, R_\Diamond, b)$. The proof for $SX(R_\Diamond, R_\Box, y)$ is similar. The proof for soundness and completeness is analogous to the proof for soundness and completeness of Theorem \ref{thm:seperation} given in Section \ref{ssuc:separation queries}. 
To prove termination, we would need the following lemma.

\begin{definition}
We define $\Diamond$-leading concepts as the smallest set of concepts  satisfying the following conditions. 
\begin{enumerate}
    \item For any atomic concept $C$, the concept $\Diamond C$ is $\Diamond$-leading.
    \item If $C$ is  $\Diamond$-leading, then  $\Diamond C$ is  $\Diamond$-leading.
    \item If $C$ is   $\Diamond$-leading, then $C\vee C_1$ is    $\Diamond$-leading for any $C_1$. 
      \item If $C_1$ and $C_2$ are  $\Diamond$-leading, then $C_1 \wedge  C_2$ is  $\Diamond$-leading. 
\end{enumerate} 
We define $\Box$-leading concepts as the smallest set of concepts  satisfying the following conditions.
\begin{enumerate}
    \item For any atomic concept $C$, the concept $\Box C$ is $\Box$-leading.
    \item If $C$ is  $\Box$-leading, then  $\Box C$ is  $\Box$-leading.
    \item If $C$ is   $\Box$-leading, then $C\wedge  C_1$ is    $\Box$-leading for any $C_1$. 
    \item If $C_1$ and $C_2$ are  $\Box$-leading, then $C_1 \vee  C_2$ is  $\Box$-leading. 
\end{enumerate} 
\end{definition}
Note that a concept $C_1 \wedge C_2$ (resp.~$C_1 \vee C_2$) is $\Diamond$-leading (resp.~$\Box$-leading) iff both $C_1$ and $C_2$ are $\Diamond$-leading (resp.~$\Box$-leading). 

\begin{lemma}\label{lem:Box-Diamond separation}
For any LE-$\mathcal{ALC}$ ABox $\mathcal{A}$ and any individual names $b$, $y$, the following holds:
\begin{enumerate}
\item If a term of the form  $\Diamond b I x_C$ or  $\Diamond b :C$ or $x_C R_\Diamond b$ or  $ b I \blacksquare x_C$ appears in $\overline{\mathcal{A}}$, then $C$ must be $\Diamond$-leading. 

\item If a term of the form  $ a_C I \Box y$ or  $\Box y::C$ or $a_C R_\Box  y$ or $\Diamondblack a_C I y$ appears in $\overline{\mathcal{A}}$, then $C$ must be $\Box$-leading. 

\item If we have a term of the form $x_C::C'$ or $a_{C'} I x_C$ or $a_{C'}:C$ in $\overline{\mathcal{A}}$,  and $C$ is $\Box$-leading, then  $C'$ is also $\Box$-leading. 

\item If we have a term of the form $a_C:C'$ or $a_{C} I x_{C'}$ or $x_{C'}::C$ in $\overline{\mathcal{A}}$,  and $C$ is $\Diamond$-leading, then  $C'$ is also $\Diamond$-leading. 
    \item No term of the form $\Diamond b I \Box y$ 
can belong to $\overline{\mathcal{A}}$. 
    \item No term of the form $\Diamond b R_\Box y$    
    can belong to $\overline{\mathcal{A}}$. 
    \item No constant of the form $\Diamondblack\Diamond b$ or $\blacksquare \Box y$ appears in $\overline{\mathcal{A}}$ for any $b$ or $y$. 
\end{enumerate}

\end{lemma}
\begin{proof}

The proof follows by a simultaneous induction on the number of applications of the expansion rules. The proof for base case is obvious as $\mathcal{A}$ does not contain individual name of the form $\Diamond b$ or $\Box y$. We give the proof for all inductive cases now.

\textbf{Creation rule:} Only terms added by this rule are of the form $a_C:C$ or $x_C::C$ for some $C \in \mathcal{A}$. For terms of both of these types,  all the items in lemma hold trivially.  

\textbf{Basic rule:} In this case, we add term $b I y$ from terms $b:C$ and $y::C$. We only need to consider the following cases: (1) $b$ is of the form $\Diamond d$, and $y$ is of the form $x_{C'}$ for some $C'$. By induction item 1, $C$ is $\Diamond$-leading. Hence, by induction item 4, $C'$ is also $\Diamond$-leading. Therefore, the new term $\Diamond d I x_{C'}$ also satisfies item 1. (2) $b$ is of the form $a_{C'}$ for some $C'$, and $y$ is of the form $\Box z$. By induction item 2, $C$ is $\Box$-leading. Hence, by induction item 3, $C'$ is also $\Box$-leading. Therefore, the new term  $a_{C'} I \Box z$ also satisfies item 2.  (3) $b$ is of the form $a_{C_1}$, and $y$ is of the form $x_{C_2}$. In this case, if $C_1$ (resp.~$C_2$) is $\Diamond$-leading (resp.~$\Box$-leading), then by induction item 4 (resp.~item 3) $C$ would be be $\Diamond$-leading (resp.~$\Box$-leading). 
By again applying the same items, we would get $C_2$ (resp.~$C_1$) is $\Diamond$-leading (resp.~$\Box$-leading). Therefore, the added term $a_{C_1} I x_{C_2}$ satisfies items 3 and 4. Item 5 is satisfied, since if any of these terms is of the form $\Diamond b I \Box y$, then both $\Diamond b:C$, and $\Box y::C$ appear in $\overline{\mathcal{A}}$. By induction items 1 and 2 $C$ must be both $\Box$-leading and ~$\Diamond$-leading. However, no such concept exists. Item 7  is satisfied as this rule does not add new individual names.

\textbf{Rules $\wedge_A$ and $\vee_X$:} We only give the proof for $\wedge_A$, the proof for $\vee_X$ is dual. In this case, we add terms $b:C_1$, and $b:C_2$ from term $b:C_1 \wedge C_2$. We need to consider the following cases: (1) $b$ is of form $\Diamond d$. By induction item 1, $C_1 \wedge C_2$ is $\Diamond$-leading. Therefore, both $C_1$ and $C_2$ must be $\Diamond$-leading. Hence, the newly added terms $b:C_1$ and $b:C_2$ also satisfy item 1. (2) $b$ is of the form $a_{C}$. We have to show items 3 and 4 hold. If $C$ is $\Diamond$-leading, then by induction item 4, $C_1 \wedge C_2$ is $\Diamond$-leading, which implies that both $C_1$ and $C_2$ are $\Diamond$-leading. Hence, the added terms $b:C_1$ and $b:C_2$ satisfy item 4. To show item 3 holds for the new  terms,  w.l.o.g.~suppose $C_1$ is $\Box$-leading. Then, by def. $C_1 \wedge C_2$ is $\Box$-leading as well. Therefore, by induction item 3, $C$ is $\Box$-leading. We can similarly show  $C$ is $\Box$-leading, when $C_2$ is $\Box$-leading.   Item 7  is satisfied as this rule does not add new individual names.

\textbf{Rules $\Box$ and $\Diamond$:}  We only give the proof for $\Box$, the proof for $\Diamond$ is dual. In this case, we add term of the form $b R_\Box y$ from terms $b:[R_\Box]C$ and $y::C$.  By induction item 6, $b$ can not be of the form $\Diamond d$. If $b$ is of the form $a_{C'}$ for some $C'$, then by induction item 3, $C'$ is $\Box$-leading. Hence, the added term $a_{C'} R_\Box y$ satisfies item 2. It also satisfies item 6 because $C'$ being $\Box$-leading can not have $\Diamond$ as the outermost connective.  Item 7  is satisfied as this rule does not add new individual names.

\textbf{Rules $\Box y$, $\blacksquare y$, $\Diamond b$, and $\Diamondblack b$:}  We only give the proof for $\Box y$, the proofs for other rules are similar. In this case, we add term of the form $b R_\Box y$ from term $b I \Box y$. By induction item 6, $b$ can not be of the form $\Diamond d$. If $b$ is of the form $a_C$ for some $C$, then by induction item 2, $C$ must be $\Box$-leading. Therefore, the added term $a_C R_\Box y$ satisfies item 2. It also satisfies item 6, since  $C$ being $\Box$-leading, can not have $\Diamond$ as the outermost connective.  Item 7  is satisfied as this rule does not add new individual names.

\textbf{Rules $\wedge_A^{-1}$ and $\vee_X^{-1}$:}  We only give the proof for $\wedge_A^{-1}$, the proof for $\vee_X^{-1}$ is dual. In this case, we add term of the form $b:C_1 \wedge C_2$ from terms of the form $b:C_1$ and $b:C_2$. We need to consider the following cases: (1) 
$b$ is of the form $\Diamond d$, then by induction item 1, $C_1$ and $C_2$ are $\Diamond$-leading. Then, by def.~$C_1 \wedge C_2$ is also $\Diamond$-leading. Therefore, the new term $b:C_1 \wedge C_2$ satisfies item 1. (2) $b$ is of the form $a_{C}$. We have to show items 3 and 4 hold. If $C$ is $\Diamond$-leading, then by induction item 4, both $C_1$ and $C_2$ are $\Diamond$-leading, which implies that $C_1 \wedge C_2$ is $\Diamond$-leading. Hence, new term $b:C_1 \wedge C_2$ satisfies item 4. To show item 3 holds for the new  term, suppose $C_1 \wedge C_2$ is $\Box$-leading. Then, $C_1$ is $\Box$-leading or $C_2$ is $\Box$-leading. Therefore, by induction item 3, $C$ is $\Box$-leading. It also satisfies item 6, since $C$ being $\Box$-leading, can not have $\Diamond$ as the outermost connective.   Item 7  is satisfied as this rule does not add new individual names.

\textbf{Rules $R_\Box$ and  $R_\Diamond$:} We only the give proof for $R_\Box$, the proof for $R_\Diamond $ is dual. In this case, we add terms $b I \Box y$, and $\Diamondblack b I y$ from $b R_\Box y$. 
By induction item 5, $b$ can not be of the form $\Diamond d$. If $b$ is of the form $a_C$ for some $C$, then by induction item 2, $C$ must be $\Box$-leading. Therefore, the added terms $a_C I \Box y$ and $\Diamondblack a_C I y$ satisfy item 2. As  $b$ can not be of the form $\Diamond d$, the possibly  new constant $\Diamondblack b$ is not of the form $\Diamondblack\Diamond d$ for any $d$. Hence, item 3 is satisfied. 

\textbf{Rules $a_C$ and $x_C$:}  We only give the proof for $a_C$, the proof for $x_C$ is dual. In this case, we add term of the form $b:C$ from term $b I x_C$. We need to consider the following cases: (1) $b$ is of the form $\Diamond d$, then by induction item 1, $C$ must be $\Diamond$-leading. Therefore, the added term $b:C$ also satisfies item 1. (2) If $b$ is of the form $a_{C'}$ for some $C'$. We have to show items 3 and 4 hold. If $C'$ is $\Diamond$-leading, then by induction item 4, $C$ is $\Diamond$-leading. Hence, added term $a_{C'}:C$ satisfies item 4. To show item 3 holds for the new  term, w.l.o.g.~suppose $C$ is $\Box$-leading. Then, by induction item 3, $C'$ is $\Box$-leading. 

\end{proof}
As a corollary, we get the following result.

\begin{corollary}\label{cor:box-diamond separation}
  For any LE-$\mathcal{ALC}$ ABox $\mathcal{A}$, and any terms $d R_\Box z$, and $y R_\Diamond b$, if  $dR_\Box z \in \overline{\mathcal{A} \cup \{y R_\Diamond b\}}$, then  $d R_\Box z \in \overline{\mathcal{A}}$.  
\end{corollary}
\begin{proof}
For any term of the form $y R_\Diamond b$, the only rule that has it in premise is the adjunction rule $R_\Diamond$ which adds terms $\Diamond b I y$, and $b I \blacksquare y$. The term $b I \blacksquare y$ cannot lead to the addition of any other term. If the term $\Diamond b I y$ leads to the addition of a term of the form $d R_\Box y$, then it means that we must have $\Diamond b :\Box C \in \overline{\mathcal{A}}$, for some $C$ or $\Diamond b I \Box z \in \overline{\mathcal{A}}$  for some $z$. However, none of these is possible by the lemma \ref{lem:Box-Diamond separation}. 
\end{proof}

This lemma immediately implies the following modified version of Lemma \ref{lem:depth}. 

\begin{lemma} 
For any individual names $b$, and $y$, and concept $C$  added during tableau expansion of ${\mathcal{A}}$, 
\begin{equation}\label{eq:IH 4}
 \Box_{\mathcal{D}}(C) \leq  \Box_{\mathcal{D}}(\mathcal{A})+1 \, 
 \mbox{ and } \, \Diamond_{\mathcal{D}}(C) \leq  \Diamond_{\mathcal{D}}(\mathcal{A})+1,  
\end{equation}

\begin{equation}\label{eq:IH 5}
-\Diamond_{\mathcal{D}}(\mathcal{A} )-2 \leq \Diamond_{\mathcal{D}} (b) \, \mbox{ and } \,   \Box_{\mathcal{D}}(b) \leq  \Box_{\mathcal{D}}(\mathcal{A})+1, 
\end{equation}

\begin{equation}\label{eq:IH 6}
-\Box_{\mathcal{D}}(\mathcal{A})-1 \leq \Box_{\mathcal{D}}(y)\, \mbox{ and } \, \Diamond_{\mathcal{D}}(y) \leq  \Diamond_{\mathcal{D}}(\mathcal{A})+2
\end{equation}
\end{lemma}
\begin{proof}

The proof proceeds by showing that the following stronger claim holds. For any tableaux expansion $ \overline{\mathcal{A}}$, obtained from $\mathcal{A}$ after any finite number of expansion steps:

1.~For any term $b I y \in \overline{\mathcal{A}}$, $\Box_{\mathcal{D}} (b) -  \Box_{\mathcal{D}} (y) \leq \Box_{\mathcal{D}}(\mathcal{A})+1$, and $ \Diamond_{\mathcal{D}} (y) -  \Diamond_{\mathcal{D}} (b) \leq \Diamond_{\mathcal{D}}(\mathcal{A})+2$.

2.~For any term $b R_\Box y \in \overline{\mathcal{A}}$, $\Box_{\mathcal{D}} (b) +1 -  \Box_{\mathcal{D}} (y) \leq \Box_{\mathcal{D}}(\mathcal{A})+1$, and $ \Diamond_{\mathcal{D}} (y) -  \Diamond_{\mathcal{D}} (b) \leq \Diamond_{\mathcal{D}}(\mathcal{A})+1$.

3.~For any term $y R_\Diamond b \in \overline{\mathcal{A}}$, $\Box_{\mathcal{D}} (b) -  \Box_{\mathcal{D}} (y) \leq \Box_{\mathcal{D}}(\mathcal{A})+1$, and $ \Diamond_{\mathcal{D}} (y) +1  -  \Diamond_{\mathcal{D}} (b) \leq \Diamond_{\mathcal{D}}(\mathcal{A})+2$.

4.~For any term $b:C \in \overline{\mathcal{A}}$, $ \Box_{\mathcal{D}} (b)  +  \Box_{\mathcal{D}} (C) \leq \Box_{\mathcal{D}}(\mathcal{A})+1$, and $ -\Diamond_{\mathcal{D}} (b) -  \Diamond_{\mathcal{D}} (C) \leq 1$.

5.~For any term $y::C \in \overline{\mathcal{A}}$, $ -\Box_{\mathcal{D}} (y)  - \Box_{\mathcal{D}} (C) \leq 0$, and $ \Diamond_{\mathcal{D}} (y) +  \Diamond_{\mathcal{D}} (C) \leq \Diamond_{\mathcal{D}}(\mathcal{A})+1$.

The proof relies on the idea that the new rule  $SA(R_\Box, R_\Diamond, b)$  introduces a new term of the form $y R_\Diamond b$ from a term $b R_\Box y$. However, by Corollary \ref{cor:box-diamond separation}
the term  $b R_\Box y$  must belong to the LE-$\mathcal{ALC}$ completion of $\mathcal{A}$. Hence, it satisfies Condition 2 by Lemma \ref{lem:depth}. The proof for all other conditions follows by a straightforward generalization of the proof of \cite[Lemma 1]{van2023non}. 
\end{proof}

\section{Models for example knowledge base}\label{App:Models for example knowledge base} 
The knowledge base in Section~\ref{sec:examples} is given by $\mathcal{K} =\mathcal{K}_1 \cup \mathcal{K}_2$, where the initial source database $\mathcal{K}_1=(\mathcal{A}_1,\mathcal{T}_1)$ is given by

\begin{multicols}{2}

\begin{align*}
    \mathcal{A}_1=\{& m_4:IM, \neg(m_4 I x_\emptyset), x_{\emptyset} :: FM\wedge IM,\\ & x_{\emptyset}::GM\wedge IM, f_1::GM, f_2::FM,  \\ & f_4::DM, m_3:RDM, m_3 I f_6, m_1 I f_3, \\ & \neg(m_1 I f_2), \neg(m_2:EUM)\},
\end{align*}
\columnbreak

\begin{align*}
    \mathcal{T}_1=\{EUM&\equiv GM\vee FM,\\
    RDM&\equiv RM\wedge DM, \\ IM&\sqsubseteq EUM\},
\end{align*}

\end{multicols}

\noindent while the database from users of two different groups $\mathcal{K}_2 = (\mathcal{A}_2,\mathcal{T}_2)$ is given by

\begin{multicols}{2}
\begin{align*}
    \mathcal{A}_2=\{& m_3 {R_\Box}_1 f_3, m_3 {R_\Box}_2 f_3,\neg(m_1 {R_\Box}_1 f_6),\\ & f_3 {R_\Diamond}_1 m_3, f_3 {R_\Diamond}_2 m_3, m_3: \Box_2 RDM, \\ & f_5::\Diamond_1 RM, \neg(m_1 {R_\Box}_2 f_5)\},
\end{align*}
\columnbreak

\begin{align*}
    \mathcal{T}_2=\{FDM &\equiv \Box_1 DM \wedge \Box_2 DM\}.
\end{align*}

\end{multicols}

By unraveling TBoxes we get the following terms:
\begin{enumerate}
    \item $EUM\equiv GM\vee FM$
    \item $RDM\equiv RM\wedge DM$
    \item $IM\equiv (GM\vee FM)\wedge C$ for some $C$ not appearing in $\mathcal{K}$
    \item $FDM \equiv \Box_1 DM \wedge \Box_2 DM$
\end{enumerate}

Note that the terms $FM\wedge IM$ and $GM\wedge IM$ in $\mathcal{A}_1$ are denoted as follows.
\begin{enumerate}
    \item $FM\wedge IM\equiv FM\wedge ((GM\vee FM)\wedge C)\equiv FM\wedge C$
    \item $GM\wedge IM\equiv GM\wedge ((GM\vee FM)\wedge C)\equiv GM\wedge C$
\end{enumerate}


We denote the objects and features in the model  of the form $a_C$, and $x_C$ as below:

\smallskip
{{
\centering
\begin{tabular}{|c|c|c|c|}
\hline
$a_1$  & $a_{GM}$ &  $x_1$  & $x_{GM}$\\
$a_2$  & $a_{FM}$ &  $x_2$  & $x_{FM}$ \\
$a_3$  & $a_{GM\vee FM}$ & $x_3$ & $x_{GM\vee FM}$  \\
$a_4$  & $a_{RM}$ &    $x_4$  & $x_{RM}$\\
$a_5$  & $a_{DM}$ &    $x_5$  & $x_{DM}$ \\
$a_6$  & $a_{RM\wedge DM}$ & $x_6$  & $x_{RM\wedge DM}$  \\
$a_7$ & $a_{C}$ & $x_7$ & $x_{C}$\\
$a_8$ & $a_{(GM\vee FM)\wedge C}$ & $x_8$ & $x_{(GM\vee FM)\wedge C}$\\
$a_9$ & $a_{FM\wedge C}$ & $x_9$ & $x_{FM\wedge C}$\\
$a_{10}$  & $a_{GM\wedge C}$ &    $x_{10}$  & $x_{GM\wedge C}$\\
$a_{11}$  & $a_{\Box_1 DM}$ &    $x_{11}$  & $x_{\Box_1 DM}$\\
$a_{12}$  & $a_{\Box_2 DM}$ &    $x_{12}$  & $x_{\Box_2 DM}$\\
$a_{13}$ & $a_{\Box_1 DM\wedge \Box_2 DM}$ & $x_{13}$ & $x_{\Box_1 DM\wedge \Box_2 DM}$\\
$a_{14}$ & $a_{\Box_2 (RM\wedge DM)}$ &    $x_{14}$  & $x_{\Box_2 (RM\wedge DM)}$\\
$a_{15}$  & $a_{\Diamond_1 RM}$ &    $x_{15}$  & $x_{\Diamond_1 RM}$\\
$a_{16}$  & $a_{\top}$ &    $x_{16}$  & $x_{\bot}$\\ 
 &  & $x_{17}$ & $x_\emptyset$\\
\hline
\end{tabular}
\par
}}
\smallskip

We give the following table depicting all the objects and features appearing in the model and whether or not they are related by $I$.  

\begin{table}[ht]
\caption{Objects ($A$) and features ($X$) of the model and Relation $I$ between them }
\label{tab:I relation} 
\centering
 \small{
 
    \begin{tabular}{|c|cccccccccccccccccc|}
\hline 
& $x_1$ & $x_2$ & $x_3$ & $x_4$ & $x_5$ & $x_6$ & $x_7$ & $x_8$ & $x_9$ & $x_{10}$ & $x_{11}$ &$x_{12}$ & $x_{13}$  & $x_{14}$ & $x_{15}$ & $x_{16}$ & $\blacksquare_1 x_{15}$ & $\Box_2 x_4$\\
\hline 
 $a_1$ & 1 & 0 & 1 & 0 & 0 & 0 &0 &0& 0& 0& 0& 0&0&0 & 0 &0 & 0 &0 \\
 $a_2$  & 0 & 1 & 1 & 0& 0 & 0 &0 &0& 0& 0& 0& 0&0&0 & 0 &0 & 0 &0\\
  $a_3$ & 0 & 0 & 1 & 0 & 0 & 0 &0 &0& 0& 0& 0& 0&0&0 & 0 &0 & 0 &0\\
 $a_4$  & 0 & 0 & 0 & 1 & 0 & 0 &0 &0& 0& 0& 0& 0&0&0 & 0 &0 & 1 &0\\
  $a_5$  & 0 & 0 & 0 & 0 & 1 & 0 &0 &0& 0& 0& 0& 0&0&0 & 0 &0 & 0 &0\\
 $a_6$   & 0 & 0 & 0 & 1 & 1 & 1&0 &0& 0& 0& 0& 0&0&0 & 0 &0 & 1 &0\\
 $a_7$  & 0 & 0 & 0 & 0 & 0 & 0 &1 &0& 0& 0& 0& 0&0&0 & 0 &0 & 0 &0\\  
 $a_8$  & 0 & 0 & 1 & 0 & 0 & 0 &1 &1& 0& 0& 0& 0&0&0 & 0 &0 & 0 &0\\   
  $a_9$  & 0 & 1 & 0 & 0 & 0 & 0 &1 &0& 1& 0& 0& 0&0&0 & 0 &0 & 0& 0 \\
 $a_{10}$   & 1 & 0 & 0 & 0 & 0 & 0 &1 &0& 0& 1& 0& 0&0&0 & 0 &0 & 0 & 0\\
  $a_{11}$  & 0 & 0 & 0 & 0 & 0 & 0 &0 &0& 0& 0& 1& 0&0&0 & 0 &0 & 0 & 0\\
  $a_{12}$  & 0 & 0 & 0 & 0 & 0 & 0 &0 &0& 0& 0& 0& 1&0& 0& 0 &0 & 0 &0\\
 $a_{13}$  & 0 & 0 & 0 & 0 & 0 & 0 &0 &0& 0& 0& 1& 1&1&0 & 0 &0 & 0& 0 \\
  $a_{14}$  & 0 & 0 & 0 & 0 & 0 & 0 &0 &0& 0& 0& 0& 1&0 &1 & 0 &0 & 0 &1\\
  $a_{15}$  & 0 & 0 & 0 & 0 & 0 & 0 &0 &0& 0& 0& 0& 0&0&0 & 1 &0 & 0 & 0\\
  $a_{16}$  & 0 & 0 & 0 & 0 & 0 & 0 &0 &0& 0& 0& 0& 0&0&0 & 0 &0 & 0 &0 \\
  
  $\Diamondblack_1 a_{11}$   & 0 & 0 & 0 & 0 & 1 & 0 &0 &0& 0& 0& 0& 0&0&0 & 0 &0 & 0 & 0\\
  $\Diamondblack_2 a_{12}$   & 0 & 0 & 0 & 0 & 1 & 0 &0 &0& 0& 0& 0& 0&0&0 & 0 &0 & 0 & 0\\
  $\Diamondblack_1 a_{13}$   & 0 & 0 & 0 & 0 & 1 & 0 &0 &0& 0& 0& 0& 0&0&0 & 0 &0 & 0 & 0\\
  $\Diamondblack_2 a_{13}$ & 0 & 0 & 0 & 0 & 1 & 0 &0 &0& 0& 0& 0& 0&0&0 & 0 &0 & 0 & 0\\
  $\Diamondblack_2 a_{14}$ & 0 & 0 & 0 & 1 & 1 & 1 &0 &0& 0& 0& 0& 0&0&0 & 0 &0 & 0 & 0\\
  $\Diamond_1 a_6$  & 0 & 0 & 0 & 0 & 0 & 0 &0 &0& 0& 0& 0& 0&0&0 & 1 &0 & 0 & 0\\

$m_1$   & 0 & 0 & 0 & 0 & 0 & 0 &0 &0& 0& 0& 0& 0&0&0 & 0 &0 & 0 & 0\\
$m_2$   & 0 & 0 & 0 & 0 & 0 & 0 &0 &0& 0& 0& 0& 0&0&0 & 0 &0 & 0 & 0\\
$m_3$   & 0 & 0 & 0 & 1 & 1 & 1 &0 &0& 0& 0& 0& 1&0&1 & 0 &0 & 1 & 1\\
$\Diamondblack_1 m_3$  & 0 & 0 & 0 & 0 & 0 & 0 &0 &0& 0& 0& 0& 0&0&0 & 0 &0 & 0 & 0\\
$\Diamondblack_2 m_3$  & 0 & 0 & 0 & 1 & 1 & 1 &0 &0& 0& 0& 0& 0&0&0 & 0 &0 & 0 & 0\\
$\Diamond_1 m_3$  & 0 & 0 & 0 & 0 & 0 & 0 &0 &0& 0& 0& 0& 0&0&0 & 1 &0 & 0 & 0\\
$\Diamond_2 m_3$  & 0 & 0 & 0 & 1 & 1 & 1 &0 &0& 0& 0& 0& 0&0&0 & 0 &0 & 0 & 0\\
$m_4$  & 0 & 0 & 1 & 0 & 0 & 0 &1 &1& 0& 0& 0& 0&0&0 & 0 &0 & 0 & 0\\

 \hline
\end{tabular}
}
\end{table}

\begin{table}[ht]
    \centering
 \small{
\begin{tabular}{|c|ccccccccccccc|}
\hline
& $f_1$ & $f_2$  & $f_3$ & $f_4$ & $f_5$  & $f_6$ & $\blacksquare_1 f_3$ & $\blacksquare_2 f_3$ & $\Box_1 f_3$ &  $\Box_2 f_3$ & $\blacksquare_1 f_5$ & $\Box_2 f_4$ & $x_{17}$\\ 
\hline
$a_1$  &1 & 0 & 0 &0 & 0 & 0 &0&0&0&0&0&0&0\\
$a_2$ &0 & 1& 0&0& 0 & 0&0&0&0&0&0&0&0\\
$a_3$ &0 & 0 & 0&0& 0 & 0&0&0&0&0&0&0&0\\
$a_4$  &0 & 0 & 0 &0& 0 & 0&0&0&0&0&1&0&0\\
$a_5$ &0 & 0 & 0&1& 0 & 0&0&0&0&0&0&0&0\\
$a_6$  &0 & 0 & 0&1& 0 & 0&0&0&0&0&1&0&0\\
$a_7$ &0 & 0 & 0&0& 0 & 0&0&0&0&0&0&0&0\\  
$a_8$ & 0 & 0 & 0&0& 0 & 0&0&0&0&0&0&0&0\\   
$a_9$ &0 & 0 & 0&0& 0 & 0&0&0&0&0&0&0&1\\
$a_{10}$ &0 & 0 & 0&0& 0 & 0&0&0&0&0&0&0&1\\
$a_{11}$ & 0 & 0 & 0&0& 0 & 0&0&0&0&0&0&0&0\\
$a_{12}$ &0 & 0 & 0&0& 0 & 0&0&0&0&0&0&0&0\\
$a_{13}$ &0 & 0 & 0&0& 0 & 0&0&0&0&0&0&0&0\\
$a_{14}$ &0 & 0 & 0&0& 0 & 0&0&0&0&0&0&1&0\\
$a_{15}$ &0 & 0 & 0&0& 1 & 0&0&0&0&0&0&0&0\\
$a_{16}$ &0 & 0 & 0&0& 0 & 0&0&0&0&0&0&0&0\\
  
$\Diamondblack_1 a_{11}$ &0 & 0 & 0 &0& 0 & 0&0&0&0&0&0&0&0\\
$\Diamondblack_2 a_{12}$ &0 & 0 & 0 &0& 0 & 0&0&0&0&0&0&0&0\\
$\Diamondblack_1 a_{13}$ &0 & 0 & 0 &0& 0 & 0&0&0&0&0&0&0&0\\
$\Diamondblack_2 a_{13}$ &0 & 0& 0&0& 0 & 0&0&0&0&0&0&0&0\\
$\Diamondblack_2 a_{14}$ &0 & 0& 0&1& 0 & 0&0&0&0&0&0&0&0\\
$\Diamond_1 a_6$  & 0 & 0 & 0 & 0 & 1 & 0 &0 &0& 0& 0& 0&0&0\\
$m_1$ & 0 & 0 & 1&0& 0 & 0&0&0&0&0&0&0&0\\
$m_2$ & 0 & 0 & 0&0& 0 & 0&0&0&0&0&0&0&0\\
$m_3$ &0 & 0 & 0 & 1 & 0 & 1&1&1&1&1&1&1&0\\
$\Diamondblack_1 m_3$ &0 & 0 & 1 & 0 & 0 & 0&0&0&0&0&0&0&0\\
$\Diamondblack_2 m_3$  &0 & 0 & 1 & 1 & 0 & 0&0&0&0&0&0&0&0\\
$\Diamond_1 m_3$  &0 & 0 & 1 & 0 & 1 & 0&0&0&0&0&0&0&0\\
$\Diamond_2 m_3$  &0 & 0 & 1 & 0 & 0 & 0&0&0&0&0&0&0&0\\
$m_4$ &0 & 0 & 0&0& 0 & 0&0&0&0&0&0&0&0\\
\hline
\end{tabular}
}

\end{table}
\smallskip

The relations ${R_\Box}_i$, and ${R_\Diamond}_i$ are given as follows. We have  ${R_\Box}_1 = \{(a_{11},x_5), (a_{13}, x_5), (m_3,f_3))\} $, and ${R_\Box}_2 = \{(a_{12},x_5), (a_{13}, x_5), (a_{14}, x_4), (a_{14}, x_5), (a_{14}, x_6),(a_{14},f_4),(m_3,x_4), (m_3,x_5), (m_3,x_6), (m_3,f_3),(m_3,f_4)\}$. ${R_\Diamond}_1 = \{(f_5,a_4),(f_5,a_6),(f_3,m_3),(f_5,m_3),(x_{15},a_6),(x_{15},a_4),(x_{15},m_3)\}$ and ${R_\Diamond}_2 = \{(f_3,m_3)\}$. The model contains atomic concepts $GM$, $FM$,  $RM$, $DM$, and $C$. For any of these concepts $D$, its interpretation is given  by the tuple $(x_D^\downarrow, a_D^\uparrow)$.

\end{document}